\documentclass{article}

\usepackage{arxiv}

\usepackage[utf8]{inputenc} % allow utf-8 input
\usepackage[T1]{fontenc}    % use 8-bit T1 fonts
\usepackage{hyperref}       % hyperlinks
\usepackage{url}            % simple URL typesetting
\usepackage{booktabs}       % professional-quality tables
\usepackage{amsfonts}       % blackboard math symbols
\usepackage{nicefrac}       % compact symbols for 1/2, etc.
\usepackage{microtype}      % microtypography
\usepackage{lipsum}
\usepackage{graphicx}
\graphicspath{ {./images/} }

\usepackage{algorithm}
\usepackage{algorithmic}
\usepackage{amsmath}
\usepackage{amsthm}
\usepackage{amssymb}
\usepackage{indentfirst}
\usepackage[inline]{enumitem}
\usepackage{romannum}
\usepackage{authblk}
\usepackage[caption=false]{subfig}
\usepackage{ifthen}
\newboolean{showcomments}
\setboolean{showcomments}{true}
\newcommand{\yu}[1]{\ifthenelse{\boolean{showcomments}}
{ \textcolor{blue}{(Yu says:  #1)}}{}}
\newcommand{\xiaoqi}[1]{\ifthenelse{\boolean{showcomments}}
{ \textcolor{red}{(Xiaoqi says:  #1)}}{}}
\newcommand{\adam}[1]{\ifthenelse{\boolean{showcomments}}
{ \textcolor{red}{(Adam says:  #1)}}{}}
\newcommand{\shai}[1]{\ifthenelse{\boolean{showcomments}}
{ \textcolor{red}{(Shai says:  #1)}}{}}
\newcommand{\shaiimp}[1]{\ifthenelse{\boolean{showcomments}}
{ \textcolor{blue}{(Shai says:  #1)}}{}}
\newcommand{\addcites}[0]{\ifthenelse{\boolean{showcomments}}
{ \textcolor{green}{(add citation(s))}}{}}
\newcommand{\addref}[0]{\ifthenelse{\boolean{showcomments}}
{ \textcolor{green}{(add ref)}}{}}
\newcommand{\todo}[1]{\ifthenelse{\boolean{showcomments}}
{ \textcolor{red}{(To do:  #1)}}{}}
\newcommand{\note}[1]{\ifthenelse{\boolean{showcomments}}
{ \textcolor{red}{#1}}{}}

\DeclareMathOperator*{\argmax}{arg\,max}
\DeclareMathOperator*{\argmin}{arg\,min}
\newcommand{\bigO}{O}
\newcommand{\loglogm}{\log m / \log\log m}
\newcommand{\opti}{\text{OPT}^{(i)}}
\newcommand{\wopti}{\text{wOPT}^{(i)}}

\newcommand{\task}[1]{#1}
\newcommand{\machine}[1]{#1}
\newcommand{\ps}[1]{s_{#1}}
\newcommand{\cs}[2]{s_{#1, #2}}
\newcommand{\nodeW}[1]{w_{#1}}
\newcommand{\edgeW}[2]{w_{#1, #2}}
\newcommand{\ga}[1]{f(#1)}
\newcommand{\gam}[1]{f_{\mathsf{mksp}}(#1)}
\newcommand{\gaw}[1]{f_{\mathsf{twct}}(#1)}
\newcommand{\ma}[1]{h(#1)}
\newcommand{\st}[1]{t(#1)}
\newcommand{\weight}[1]{\omega_#1}
\newcommand{\tl}{q}
\newcommand{\tlu}{Q}
\newcommand{\js}{\mathcal{J}}

\newcommand{\tc}{\mathbb{C}}

\newcommand{\sbar}[2]{\bar{s}(c_{#1}, c_{#2})}
\newcommand{\ebar}[1]{\bar{e}_{#1}}
\newcommand{\e}[1]{{e}_{#1}}

\newtheorem{theorem}{Theorem}[section]
\newtheorem{proposition}[theorem]{Proposition}

\newtheorem{lemma}[theorem]{Lemma}

\newtheorem{example}{Example}

\title{Communication-Aware Scheduling of Precedence-Constrained Tasks on Related Machines}

\author[\space\space 1]{Yu Su\thanks{suyu@caltech.edu}}
\author[\space\space 2]{Xiaoqi Ren\thanks{xiaoqiren@google.com}}
\author[\space\space 3]{Shai Vardi\thanks{svardi@purdue.edu}}
\author[\space\space 1]{Adam Wierman\thanks{adamw@caltech.edu}}
\affil[1]{Department of Computing and Mathematical Sciences, California Institute of Technology, Pasadena, CA}
\affil[2]{Google, Kirkland, WA}
\affil[3]{Krannert School of Management, Purdue University, West Lafayette, IN}
\date{}
\begin{document}

\pagenumbering{arabic}
\maketitle

\begin{abstract}
Scheduling precedence-constrained tasks is a classical problem that has been studied for more than fifty years. However, little progress has been made in the setting where there are communication delays between tasks.  Results for the case of identical machines were derived nearly thirty years ago, and yet no results for related machines have followed. In this work, we propose a new scheduler, Generalized Earliest Time First (GETF), and provide the first provable, worst-case approximation guarantees for the goals of minimizing both the makespan and total weighted completion time of tasks with precedence constraints on related machines with machine-dependent communication times.  %This scheduling problem has been originally considered in the domain of parallel processing, but recently has attracted lots of attention due to its applications in large-scale machine learning platforms. 
\end{abstract}

% keywords can be removed

\section{Introduction}
In this paper we study scheduling precedence-constrained tasks onto a set of heterogeneous machines with communication delays between the machines in order to minimize the makespan or the total weighted completion time.
% \shai{isn't this two problems? Also, I am not sure how classical this very specific problem is and unless Coffman and Bruno studied this exact problem (or more precisely two problems, this might put the reviewer off before they get past the first sentence. I suggest: In this paper we study scheduling precedence-constrained tasks onto a set of heterogeneous machines with communication delays between the machines in order to minimize the makespan or the total weighted completion time.}
Initially, work on this topic was motivated by the goal of scheduling jobs on multi-processor systems, e.g., \cite{coffman1976computer}. Today this problem is timely due to the prominence of large-scale, general-purpose machine learning platforms.  For example, in systems such as Google's TensorFlow \cite{abadi2016tensorflow}, Facebook's PyTorch \cite{paszke2017automatic} and Microsoft's Azure Machine Learning (AzureML) \cite{chappell2015azure}, machine learning workflows are expressed via a computational graph, where jobs are made up of tasks, represented as vertices, and precedence relationships between the tasks, represented as edges.  This ``precedence graph'' abstraction allows data scientists to quickly develop and incorporate modular components into their machine learning pipeline (e.g., data preprocessing, model training, and model evaluation) and then easily specify a workflow. %  Sometimes this design phase is done using a graphical environment, e.g., in AzureML.    
The graphs that specify the workflows in platforms such as TensorFlow, PyTorch and AzureML can be made up of hundreds or even thousands of tasks, and the jobs may be run on systems with thousands of machines.  As a result, the performance of the platforms depends on how these precedence-constrained tasks are scheduled across machines. 
% \adam{Maybe add a note that about the batch case being relevante (as opposed on online) here?}%  This scheduling problem is known to be very challenging:  both the goal of  partitioning the jobs across machines and of scheduling the jobs among a group of machines are NP-hard \cite{mayer2017tensorflow}. Hence,  one can only hope to develop approximation algorithms and/or heuristics given the scale of the problems faced by production systems.

The goal of scheduling jobs composed of  precedence-constrained tasks has been studied for more than fifty years, starting with the work of \cite{graham1969bounds}.  The simplest version of this scheduling problem focuses on scheduling a single job with $n$ precedence-constrained tasks on $m$ identical parallel machines with the goal of minimizing the \textit{makespan}: the time until the last task completes. More generally, the goal of minimizing the \textit{total weighted completion time} is considered, where the total weighted completion time is a weighted average of the completion time of each task in the job\footnote{Makespan is a special case of total weighted completion time as a dummy task with weight one can be added as the final task of the job, with all other tasks given weight zero.}. For the goal of minimizing the makespan, Graham showed that a simple list scheduling algorithm can find a schedule of length within a multiplicative factor of $(2-1/m)$ of the optimal. %, i.e., that list scheduling is a $(2-1/m)$-approximation.
This result is still the best guarantee known for this simple setting. %\shai{the reader will probably wonder whether it is known to be optimal, and if no, what lower bound is known} \yu{we covered these later in the literature review, but we can briefly mention it here if necessary.} 
Since then, research has sought to generalize the setting considered in two important ways: (i) to non-identical machines and (ii) to the case where communication is needed between tasks. 

% \shai{use "cite" if you say "this result was extended by" and citep if you say  "this result was extended in"}

Addressing these two issues has been one of the major goals of the field since Graham's initial result fifty years ago. Since that time, considerable progress has mostly been made on generalizations to heterogeneous machines.  The focus has been on \textit{(uniformly) related machines}, a model where each machine $i$ has a speed $s_i$, each task $j$ has a size $\nodeW{j}$, and the time to run task $j$ on machine $i$ is $\nodeW{j}/s_i$.  Under the related machine model, a sequence of results in the 1980s and 1990s culminated in a result that showed how to use list scheduling algorithms in combination with a partitioning of machines into groups with ``similar'' speeds in order to achieve an $O(\log m)$-approximation algorithm for makespan \cite{chudak1999approximation}. This result was also extended in the same work to total weighted completion time by proposing a time-indexed linear programming technique.  The extension yields an $O(\log m)$-approximation for total weighted completion time. The idea of using a \textit{group assignment} rule to partition machines into groups of machines with similar speeds and then to assign tasks to a group is a powerful one and has shown up frequently in the years since; it recently led to a breakthrough when the idea of partitioning machines was adapted further and combined with a variation of list scheduling to obtain a $\bigO(\loglogm)$-approximation algorithm for both makespan and total weighted completion time \cite{li2017scheduling}.

Despite the progress made in generalizing from identical machines to heterogeneous machines, there has been little progress toward the goal of incorporating communication delays. Machine-dependent communication delays are crucial for capturing issues such as data locality and the difference between intra-rack and inter-rack communication.  We note that if communication  delays are machine independent, they can simply be viewed as part of the processing time, making the problem much easier.  The state-of-the-art result in the case of communication delays is~\cite{hwang1989scheduling}, which studies machine-dependent communication costs in the setting of \textit{identical machines}.  In this context, a greedy algorithm called Earliest Time First (ETF) has been shown to produce schedules with a makespan bounded by $(2-1/m)\opti + C$, where $\opti$ is the optimal schedule length when ignoring communication time and $C$ is the maximum amount of communication of a chain (path) in the precedence graph.  However, the analysis for the case of identical machines in \cite{hwang1989scheduling} is quite complex and it has proven difficult to generalize to the related machines setting.  As a result, there has been no progress outside the context of identical machines in the thirty years since \cite{hwang1989scheduling}. %\shai{but has there been progress on identical machines? This sentence seems to imply there has been} \yu{yes, there has been progress in generalizing the result from homogeneous settings to heterogeneous settings in the case of identical machines, as mentioned in the previous paragraph.}

Given the challenge of designing schedulers that are approximately optimal for related machines with machine-dependent communication time, most work studying the design of scheduling policies in this context has relied on developing scheduling heuristics and evaluating these heuristics numerically, e.g., \cite{ananthanarayanan2014grass, ren2015hopper, wu1990hypertool, xu2015hybrid, yang1994dsc, topcuoglu2002performance}. For a recent survey see \cite{wu2015workflow, mayer2017tensorflow} and the references therein.  

\textbf{Contributions.} In this paper we propose a new scheduler, Generalized Earliest Time First (GETF), and prove that it computes a makespan that is at most of length $\bigO(\loglogm) \opti + C$ in the case of related machines and machine-dependent communication times, where $C$ is the amount of communication time in a chain (path) in the precedence graph. Additionally, we generalize our result to the objective of total weighted completion time and show that GETF produces a schedule $\mathcal{S}$ whose total weighted completion time is at most $\bigO(\loglogm)$ $\wopti + \sum_j \weight{j}C(\mathcal{S},j)$, where $\wopti$ is the optimal total weighted completion time, $\weight{j}$ is the weight in the objective, and $C(\mathcal{S},j)$ is the communication requirement in a chain in the precedence graph. These two results address long-standing open problems. Note that the makespan result matches state-of-the-art bounds for the special cases (i) when there is zero communication time and (ii) when the machines are identical. In the case of total weighted completion time, no previous result exists for the case of identical machines with communication time, but the result matches the best known bound for the case with related machines and zero communication time.%\shai{next sentence repeats almost verbatim what has already been said. }  Specifically, the general problem of how to schedule precedence-constrained tasks on related machines with machine-dependent communication times has seemed unapproachable, with minimal progress since \cite{hwang1989scheduling}.

The key technical advance that enables our new result is a dramatically simplified analysis of ETF in the setting of identical machines. The state-of-the-art result in this setting is \cite{hwang1989scheduling}, which is established using a long, complex argument. In contrast, the core idea in our proof of Theorem \ref{theorem:main} is a short, simple proof of a \textit{Separation Principle} which can be used to provide a novel proof of the approximation ratio for ETF in the case of identical machines.  The proof is simple and general enough that it can be extended from identical machines to related machines by adapting recent advances from \cite{li2017scheduling}.  

\textbf{Related literature.}
%The work in this paper connects to two sets of literature: (i) research focused on improving the scheduling of tasks within machine learning platforms and (ii) research focused on the design and analysis of scheduling precedence-constrained tasks across parallel machines.  The former, includes a variety of system design and scheduling heuristics, while the later focuses on developing algorithms with provable approximation ratios. \shai{from everything we've said so far, this seems like an attempt to sell the paper, but I definitely only buy ii) as truly related despite the fact that we can use it for i)... maybe just remove this paragraph?}
In recent years, the design and optimization of large-scale general-purpose machine learning platforms has been an overarching goal, bridging many communities in both industry and academia.  The emergence of platforms such as TensorFlow, PyTorch and AzureML illustrate the power of such systems to democratize tools from machine learning, making them accessible and scalable for anyone.

Since the emergence of such systems, there has been a torrent of work that seeks to optimize the scheduling and assignment of the precedence-constrained graphs in such systems.  Heuristics have emerged for managing straggler tasks, e.g., \cite{ren2015hopper, ananthanarayanan2013effective, ananthanarayanan2014grass, ananthanarayanan2010reining}; scheduling tasks with different computational properties, e.g., jobs with MapReduce-type structures \cite{vavilapalli2013apache, lin2013joint, palanisamy2011purlieus, tan2012delay, verma2012two, wang2016maptask}, scheduling approximation jobs \cite{ananthanarayanan2014grass, zaharia2008improving, ananthanarayanan2010reining}, and managing communication times \cite{mayer2017tensorflow, hashemi2018tictac}.  Many of these heuristics have led to system designs that have had a significant industrial impact.  

Such designs typically address the challenges associated with precedence constraints in ad hoc ways based on simplifying assumptions about the structures of the graphs.
In contrast, there is a long history of analytic work seeking to design schedulers for precedence-constrained tasks with provable worst-case guarantees.  As we have already mentioned, the initial results on this topic for makespan were provided by Graham, who gave a $(2-1/m)$-approximation algorithm based on list scheduling for $P |prec|C_{max}$ \cite{graham1969bounds}. A decade later, it was shown by \cite{lenstra1978complexity} that it is NP-hard to approximate $P |prec|C_{max}$ within a factor of $4/3$. This left a gap which has been essentially closed recently, when \cite{svensson2010conditional} proved that it is NP-hard to achieve an approximation factor less than $2$, given the assumption of a new variant of the Unique Game Conjecture introduced by \cite{bansal2009optimal}. In the case of total weighted completion time objective $P | prec | \sum_j \weight{j} C_j$, the negative results carry over from the makespan objective since makespan objective can be viewed as a special case of total weighted completion time objective. Moreover, under the assumption of the stronger version of the Unique Game Conjecture, it is shown in \cite{bansal2009optimal} that it is even hard to approximate within a factor of $2-\epsilon$ for the problem with one machine. On the positive side, a $7$-approximation was given in \cite{hall1996scheduling}, and
\cite{munier1998approximation,queyranne2006approximation} later improved it to a $4$-approximation. The current best known result is a $(2+2\ln{2}+\epsilon)$-approximation by \cite{li2017scheduling} via a time-indexed linear programming relaxation technique.

The results mentioned above all focus on identical machines with zero communication delays.  When related machines are considered, the problem becomes more challenging.  An early result on this topic is \cite{chudak1999approximation}, which proposed a Speed-based List Scheduling (SLS) algorithm that obtains an approximation of $\bigO(\log{m})$ for $Q|prec|C_{max}$. A time-indexed linear programming technique has been proposed in the same work that gives a $\bigO(\log{m})$ bound for $Q|prec|\sum_j \weight{j}C_j$. Recently, an improvement to $\bigO(\loglogm)$ for both objectives was proven in \cite{li2017scheduling}. The best known lower bound for the problem of related machines is from \cite{bazzi2015towards}, which shows that it is impossible for a polynomial time algorithm to approximate the minimal makespan to any constant factor assuming the hardness of an optimization problem on $k$-partite graphs.

In contrast, when communication delay is considered, much less is known. To our knowledge, no approximation ratio is known for $P | prec, c_{i,j} |C_{max}$, and this open problem was noted by \cite{drozdowski2009scheduling}. The only algorithm with a guaranteed worst-case performance bound in this setting is ETF \cite{hwang1989scheduling}, which provides a bound of $(2-1/m)\opti + C$ on the makespan in the case of identical machines. Prior to our paper, no algorithm with a worst-case approximation guarantee for either makespan or total weighted completion time is known for the case of related machines with communication delays, i.e., $Q | prec, c_{i,j} |C_{max}$ and $Q | prec, c_{i,j} |\sum_j \weight{j}C_j$.

\section{Problem formulation} \label{section:model}
We study a model that generalizes $Q|prec,c_{i,j}|\sum_j \weight{j} C_j$ by  including machine-dependent communication times. Our goal is to derive bounds on the total weighted completion time and the  makespan, which is an important special case of the total weighted completion time that uses a particular choice of $\weight{j}$.%\shai{a special case thereof.} \yu{I prefer the longer expression as I find it more clear though lengthy?} 

Specifically, we consider the task of scheduling a job made up of a set $V$ of $n$ tasks on a heterogeneous system composed of a set $M$ of $m$ machines with potentially different processing speeds and communication speeds.  The tasks form a directed acyclic graph (DAG) $G = (V, E)$, in which each node $\task{j}$ represents a task and an edge $(\task{j'}, \task{j})$ between task $\task{j}$ and task $\task{j'}$ represents a precedence constraint. We interchangeably use node or task, as convenient. Precedence constraints are denoted by a partial order $\prec$ between two nodes of any edge, where $\task{j'} \prec \task{j}$ means that task $\task{j}$ can only be scheduled after task $\task{j'}$ completes. Let $\nodeW{j}$ represent the processing demand of task $\task{j}$. The amount of data to be transmitted between task $\task{j'}$ and task $\task{j}$ is represented by the edge weight $\edgeW{j'}{j}$ of $(\task{j'}, \task{j})$.

The system is heterogeneous in two aspects: processing speed and communication speed. For processing speed, we consider the classical \textit{related machines} model: a machine $\machine{i}$ has speed $\ps{i}$, and it takes $\nodeW{j}/\ps{i}$ uninterrupted time units for task $\task{j}$ to complete on machine $\machine{i}$. Specifically, computer resources such as CPUs and GPUs have varying speeds; hence schedulers must be able to handle heterogeneous servers. The communication speed $\cs{i'}{i}$ between any two machines $\machine{i'}, \machine{i}$ is heterogeneous across different machine pairs.  We index the machine to which task $\task{j}$ is assigned by $\ma{j}$. If $i = \ma{j}$ and $i' = \ma{j'}$, then communication time between task $\task{j'}$ and $\task{j}$ in the DAG is $\edgeW{j'}{j} / \cs{i'}{i}$. 
% Machine dependent communication times are essential to the formulation. If communication times were machine independent, they could have been viewed as part of the processing time, making the problem much easier. However, machine dependent communication times have not been previously considered in the literature on scheduling precedence-constrained tasks, but they are crucial for capturing issues such as data locality and the difference between intra-rack and inter-rack communication. \shai{I agree with Adam's addition, but the paragraph needs to be reworked as it reads a bit odd now... MOVED TO INTRO - REMOVE FROM HERE}

For simplicity, we consider a setting where the machines are fully connected to each other, so any machine can communicate with any other machine. This is without loss of generality as one can simply set the communication speed between any two disconnected machines to 0. We also assume that the DAG is connected. Again, this is without loss of generality because, otherwise, the DAG can be viewed as multiple DAGs and the same results can be applied to each. As a result, our results trivially apply to the case of  multiple jobs.
Additionally, our model assumes that each machine (processing unit) can process at most one task at a time, i.e., there is no \textit{time-sharing}, and the machines are assumed to be \textit{non-preemptive}, i.e., once a task starts on a machine, the scheduler must wait for the task to complete before assigning any new task to this machine.  This is a natural assumption in many settings, as interrupting a task and transferring it to another machine can cause significant processing overhead and communication delays due to data locality, e.g., \cite{kwok1999static}.
% For example, in machine learning platforms, interrupting a task and transferring it to another machine can cause significant processing overhead and communication delays due to data locality \citep{kwok1999static}. \shai{maybe just : This is a natural assumption in many settings, as interrupting a task and transferring it to another machine can cause significant processing overhead and communication delays due to data locality, e.g., \citep{kwok1999static}. }

The goal of the scheduler in our model is to minimize the \textit{total weighted completion time} of the job, denoted by $\sum_j \weight{j} C_j$, where $C_j$ is the completion time of task $j$ and $\weight{j}$ is the weight associated with task $j$. We also consider the \textit{makespan}, denoted by $C_{max}$, which is the time when the the final task in the DAG completes.
Note that the problem we consider is an offline scheduling problem.  This is a classical problem with relevance to modern ML platforms, which use batch scheduling of precedence constrained tasks in their pipelines, e.g. \cite{abadi2016tensorflow}.  It is also known to be challenging. Specifically, minimizing the makespan (and hence also minimizing the total weighted completion time) of jobs with precedence constraints is known to be NP-complete \cite{garey1979computers}. Thus, we aim to design a polynomial-time algorithm that computes an \textit{approximately} optimal schedule. We say that an algorithm is a $\rho$-approximation algorithm if it always produces a solution with an objective value within a factor of $\rho$ of optimal in polynomial time.

Our main results use three important concepts. First, our results provide bounds in terms of $\opti$ and $\wopti$, which are  the optimal makespan and the optimal total weighted completion time if the communication delays were zero, respectively. Note that $\opti$ and $\wopti$ are a lower bound of the corresponding objectives of the problem when communication delays are not included. Second, we provide bounds in terms of the communication time of a \textit{terminal chain} of the schedule. A \textit{chain} in the DAG is a sequence of immediate predecessor-successor pairs, whose first node is a node with no predecessor and last node is a leaf node with no successors. Third, we provide bounds in terms of the communication time of a terminal chain of a subset of the DAG that is naturally formed in the scheduling process. %\shai{I am left unimpressed by this paragraph, and don't understand the point of it} 
Formally, for any given schedule, a terminal chain $\tc$ of length $N$ can be constructed in the following fashion. We start with one of the tasks that ends last in the given schedule, denoted as $\task{c_N}$. Among all the immediate predecessors of node $\task{c_N}$, we pick one of the tasks that finishes last and define it as $\task{c_{N-1}}$. In such a way, we can construct a chain of tasks $\task{c_1} \prec \task{c_2} \prec \ldots \prec \task{c_N}$ until the first node $\task{c_1}$ in the chain does not have a predecessor. There may be many such terminal chains, and our results apply to any arbitrary terminal chain for the given schedule.
%The model we consider in this paper was initially used in the context of scheduling multiprocessor systems \cite{coffman1976computer}; however, it has recently attracted attention because of the emergence of large machine learning platforms where workflows are specified via precedence graphs. In both contexts, computing resources such as CPUs and GPUs have varying speeds; hence schedulers must be able to handle heterogeneous servers.  Similarly, in both contexts, data placement makes the communication overhead in between tasks an important consideration for the scheduler.  Thus, it is crucial for a scheduler to be able to apply in the context of related machines with machine-dependent communication, as we consider here. \shai{this paragraph is completely out of place and belongs in the intro... As it repeats things that were in the intro, I suggest deleting it and adding anything that is here but not in the intro to the intro.}

\section{Generalized Earliest Time First (GETF) Scheduling} \label{section:GETF}
In this section, we introduce a new algorithm -- Generalized Earliest Time First (GETF) -- for scheduling tasks with precedence constraints in settings where servers have heterogeneous service rates and communication times. For GETF, we provide provable worst-case approximation guarantees for both the goal of minimizing the makespan and minimizing the total weighted completion time. 

At its core, GETF is a greedy algorithm.  Like ETF, it seeks to run tasks that can be started earliest, thus minimizing the idle time created by the precedence constraints in a greedy way.  However, this simple heuristic does not take into account the potential difference between the service rates of different machines.  For this, GETF is similar to SLS. It uses a group assignment function $\ga{\cdot}$ to determine sets of ``similar'' machines and then assigns tasks to different groups of machines. Within the groups of similar machines, GETF uses the ETF greedy allocation rule. 

%\shai{weird paragraph structure, jumping between ideas}
GETF is parameterized by a group assignment function $\ga{\cdot}$ and a tie-breaking rule, and proceeds in two stages. At every iteration, GETF finds a set $A$ of all the tasks that are ready to process and are not yet scheduled. For every task in $A$, GETF calculates the earliest starting time if it was only allowed to schedule on machines in the assigned group. Then, GETF computes $B$, the set of tasks in $A$ with the earliest starting times, and chooses one of the tasks to process on a machine based on the tie-breaking rule. The pseudocode for GETF is presented in Algorithm \ref{algo:generalized_ETF} and Figure \ref{diagram:etf} in section \ref{section:comparison} illustrates the operation of GETF on a simple example (Example \ref{example:etf}). 

\begin{algorithm}[t]
\caption{ Generalized Earliest Time First (GETF) }\label{algo:generalized_ETF}
\begin{flushleft}
        \textbf{INPUT:} group assignment rule $\ga{\cdot}$, tie-breaking rule\\
        \textbf{OUTPUT:} schedule $\mathcal{S}$ with machine assignment mapping $\ma{\cdot}$ and starting time mapping $\st{\cdot}$
\end{flushleft}
\begin{algorithmic}[1]
\STATE $R \leftarrow \{\task{1}, \task{2}, \ldots, \task{n}\}$
\WHILE{$R \neq \emptyset$ }
    \STATE $A = \{ \task{j}: \task{j} \in R, \nexists j' \text{ s.t. } \task{j'} \in R \text{ and } \task{j'} \prec \task{j} \}$
    \STATE $\text{For }\task{j} \in A, t_j' = \text{ earliest starting time on machine } m_j'$ s.t. ${m_j'} \in \ga{j}$
    \STATE $B = \{ \task{j}: j = \argmin_{\task{j'} \in A} \st{j'} \}$
    \STATE Choose $\task{j}$ from $B$ to start on machine $m_j'$ with a starting time $t_j'$ based on the given tie-breaking rule
    \STATE $\ma{j} = m_j', \st{j} = t_j'$
    \STATE $R \leftarrow R \setminus \{\task{j} \}$
\ENDWHILE
\end{algorithmic}
\end{algorithm}

GETF can be instantiated with different group assignment and tie-breaking rules. To understand how these rules work, consider a situation where the $m$ machines are divided into $K$ groups $M_1, M_2, \ldots, M_K$ by a group assignment rule. Let $\ga{j}$ denote the group of machines to which task $\task{j}$ can be assigned, $j = 1, \ldots, n$. Given this notation, a schedule under GETF consists of two mappings: a mapping $\ma{\cdot}$ from each task to its assigned machine and a mapping $\st{\cdot}$ from each task to its starting time.  Further, for any schedule with $\ma{\cdot}$ produced by GETF, $\ma{\cdot}$ of the produced schedule should be consistent with group assignment function $\ga{\cdot}$, i.e., $\ma{j} \in \ga{j}$ for each task $\task{j}$.%\shai{in general I'm not sure what this paragraph adds, other than notation, but it's not phrased in a way to reflect that it's introducting notation...}

The choice of the group assignment rule has a significant impact on the performance of GETF. Indeed, different group assignment functions are used for the goals of minimizing the makespan and total weighted completion time. While our results hold for any tie-breaking rule, different tie-breaking rules could provide meaningful improvements in real-world workloads. As it could be helpful to keep a specific tie-breaking rule in mind while considering the algorithm and proofs, the reader may find it helpful to consider random tie-breaking. Our technical results are based on the specific group assignment functions described in the following subsections.

\subsection{A Group Assignment Rule for Makespan} \label{section:group_assign_makespan}

The group assignment rule $\gam{\cdot}$ for the goal of minimizing the makespan that we focus on is adapted from SLS, which is designed for the setting \textit{without} communication time.  Specifically, machines of similar speeds are grouped together as follows. 

First, all the machines with speed less than a $\frac{1}{m}$ fraction of the speed of the fastest machine are discarded.  Then, the remaining machines are divided into $K$ groups $M_1, M_2, \ldots, M_K$ where $K = \lceil \log_\gamma m \rceil$, $\gamma = \loglogm$. Note that $K = \bigO(\loglogm)$. Given the removal of the slowest machines, we can assume that any remaining machine has speed within a factor of $\frac{1}{m}$ of the fastest machine. Without loss of generality, we assume the speed of the fastest machine is $m$ and the group $M_k$ contains machines with speeds in range $[\gamma^{k-1}, \gamma^k)$.

It may seem strange that some machines are discarded, but note that the total speed of discarded machines is not bigger than the speed of the fastest machine. So, if we consider the scheduling problem with zero communication time, removing these machines at most doubles the makespan in the worst case. 

After dividing machines into $K$ groups in the preprocessing step, we need to assign the machines. This step is more involved than the division. The design of the group assignment rule $\gam{\cdot}$ is based on the solution of a linear program (LP), which is a relaxed version of the following mixed integer linear program (MILP).
\begin{subequations}\label{opt:milp}
\begin{alignat}{3}
\min_{x_{i, j}, C_j, T} & \quad T \notag \\
\sum_{i} x_{i,j} & = 1 & \forall j \label{cons:assigned}\\
\nodeW{j} \sum_{i}  \frac{x_{i, j}}{\ps{i}} & \leq C_j & \forall j \label{cons:cap}\\
C_{j'} + \nodeW{j} \sum_{i}  \frac{x_{i, j}}{\ps{i}}  & \leq C_j & j^\prime \prec j \label{cons:prec}\\
\frac{1}{\ps{i}} \sum_{j} \nodeW{j} x_{i,j} & \leq T & \forall i \label{cons:cap_each_machine}\\
C_j & \leq T & \forall j \label{cons:completion_time}\\
x_{i, j} & \in \{0, 1\} & \forall i, j \label{cons:relaxed}
\end{alignat}
\end{subequations}

While the MILP is only designed to produce a group assignment rule, its optimal solution does not necessarily provide a feasible schedule. In the MILP, $x_{i, j} = 1$ if task $\task{j}$ is assigned to machine $\machine{i}$; otherwise $x_{i, j} = 0$. For each task $\task{j}$, $C_j$ denotes the completion time of task $\task{j}$. %Let $T$ denote the objective \shai{already defined in the MILP, also why is this here in any case? we don't refer to T again in this subsection} \yu{I agree that it is quite clear what T is in the MILP, but don't we need to explain in words what T is?}. 
Constraint \eqref{cons:assigned} ensures that every task is processed on some machine. For any task $\task{j}$, processing time $\nodeW{j} \sum_{i}  \frac{x_{i, j}}{\ps{i}}$ is bounded by its completion time as in constraint \eqref{cons:cap}. Constraint \eqref{cons:prec} enforces the precedence constraints between any predecessor-successor pair $(\task{j'}, \task{j})$. 
Constraint \eqref{cons:cap_each_machine} guarantees that the total load assigned to machine $\machine{i}$ is $\nodeW{j} \sum_{i}  \frac{x_{i, j}}{\ps{i}}$ and it should not be greater than the makespan.
% \shai{I prefer to have simply "Constraint 1d) guarantees that for all of them, instead of putting some at the end of the sentence, which makes things unnecessarily difficult for the reader}
Finally, constraint \eqref{cons:completion_time} states that the makespan should not be smaller than the completion time of any task.

Since we cannot solve the MILP efficiently, we relax it to form an LP by replacing constraint \eqref{cons:relaxed} with $x_{i, j} \geq 0$. Let $x^*, C^*, T^*$ denote the optimal solution of this LP. Note that $T^*$ provides a lower bound on $\opti$, the optimal makespan for the same problem with zero communication time.

For a set $M_k \subseteq M$ of machines, let $s(M_k)$ denote the total speed of machines in $M_k$, i.e., 
\begin{equation*}
    s(M_k) = \sum_{\machine{i} \in M_k} \ps{i}.
\end{equation*}
Define $x_{M_k, j}^*$ as the total fraction of task $\task{j}$ assigned to machines in set $M_k$:
\begin{equation*}
    x_{M_k, j}^* = \sum_{i \in M_k} x^*_{i, j}.
\end{equation*}
For any task $\task{j}$, define $\ell_j$ as the largest group index such that at least  half of the tasks are fractionally assigned to machines in groups $M_\ell, \ldots, M_K$: 
\begin{equation*}
    \ell_j = \max_\ell \ell \quad \text{s.t.} \quad \sum_{k=\ell}^K  x^*_{M_k, j} \geq \frac{1}{2}.
\end{equation*}
We note that any choice of constant above works for the purpose of our worst case analysis of GETF, but the choice can potentially have an impact on its empirical performance. Thus the choice of the parameter should be further optimized when applied in practice.
% \shai{change "Note" to "we note" or something, as we haven't analysed it yet, and the reader can't note it yet}
Each task $\task{j}$ is assigned to the group $\gam{j}$ that maximizes the total speed of machines in that group among candidates $M_{l_j}, \ldots, M_K$, i.e.,
\begin{equation*}
    \gam{j} = \argmax_{M_k :\ell_j \leq k \leq K} \quad s(M_k).
\end{equation*}

\subsection{A Group Assignment Rule for Total Weighted Completion Time}
The group assignment rule $\gaw{\cdot}$ for the goal of minimizing the total weighted completion time is similar in spirit to $\gam{\cdot}$ but is based on modified solutions of a different LP.  We divide machines into groups in the same way as in Section \ref{section:group_assign_makespan}. Without loss of generality, we assume that $\frac{\nodeW{j}}{\ps{i}} \geq 1$ for any task $\task{j}$ to be processed on any machine $\machine{i}$. Thus, we can divide the time horizon into the following time-indexed intervals of possible task completion times: $[1, 2], (2, 4], (4, 8], \ldots, (\tau_{\tlu-1}, \tau_\tlu]$ where $Q = \log{(\sum_j \frac{\nodeW{j}}{\min_i \ps{i}})}$ and $\tau_\tl = 2^\tl$ for $0 \leq \tl \leq \tlu$. Then, the MILP that forms the basis for the group assignment rule can be formulated as follows:
\begin{subequations}\label{weighted:milp}
\begin{alignat}{3}
\min_{x_{i, j, \tl}, C_j} \quad \sum_j \weight{j} C_j \notag \\
\sum_{i} \sum_{\tl} x_{i,j, \tl} & = 1 & \forall j \label{cons:done} \\
\nodeW{j} \sum_{i} \frac{1}{\ps{i}} \sum_\tl x_{i, j, \tl}  & \leq C_j & \forall j \label{cons:weighted_cap}\\
C_{j'} + \nodeW{j} \sum_{i} \frac{1}{\ps{i}} \sum_\tl x_{i, j, \tl}  & \leq C_j & j^\prime \prec j \label{cons:weighted_prec}\\
\sum_{t=1}^\tl \sum_i x_{i,j,t} - \sum_{t=1}^\tl \sum_i x_{i,j^\prime,t} & \leq 0  & \forall \tl, j^\prime \prec j \label{cons:weighted_prec_all}\\
\sum_\tl \tau_{\tl-1} \sum_i x_{i,j,\tl} & < C_j & \forall j \label{cons:left_index}\\
% \sum_\tl \tau_{\tl} \sum_i x_{i,j,\tl} & \geq C_j & \forall j \label{cons:right_index}\\
\frac{1}{\ps{i}} \sum_j \nodeW{j} \sum_{t=1}^\tl x_{i,j,t} & \leq \tau_\tl & \forall i, \tl \label{cons:index_cap}\\
x_{i, j, \tl} & \in \{0, 1\} & \forall i, j, \tl \label{cons:binary}
\end{alignat}
\end{subequations}

Again, the MILP is only designed to find a group assignment rule and thus its optimal solution does not necessarily produce a feasible schedule. Here, $x_{i,j,\tl} = 1$ if task $\task{j}$ is assigned to machine $\machine{i}$ and it completes in the $\tl$th interval $(\tau_{\tl-1}, \tau_\tl]$. For each task $\task{j}$, $C_j$ denotes the completion time of task $\task{j}$ and $\weight{j}$ represents its weight in the objective of total weighted completion time. Constraint \eqref{cons:done} enforces that each task will be assigned to some machine. Constraint \eqref{cons:weighted_cap} guarantees that the completion time of a task is not smaller than its processing time. Constraints~\eqref{cons:weighted_prec} and~\eqref{cons:weighted_prec_all} together enforce the precedence constraint for every predecessor-successor pair. Constraint \eqref{cons:left_index} guarantees that the completion time of task $\task{j}$ is not  smaller than the left boundary of the $\tl$th interval $(\tau_{\tl-1}, \tau_\tl]$. The total load assigned to machine $\machine{i}$ up to $\tl$th interval is $\frac{1}{\ps{i}} \sum_j \nodeW{j} \sum_{t=1}^\tl x_{i,j,t}$, and  it should not be greater than the upper bound $\tau_\tl$ as enforced in constraint \eqref{cons:index_cap}.
% \shai{much better than the previous, except the last line}

To define the group allocation rule, we relax  constraint \eqref{cons:binary} to form an LP. As in the previous section, let $x^*, C^*$ denote the optimal solution for this LP. Note that $\sum_j \weight{j} C^*_j$ provides a lower bound for $\wopti$. For any task $\task{j}$, define $\tl(j)$ as the the minimum value of $\tl$ such that both
$\sum_{t=1}^{\tl} \sum_i x^*_{i,j,t} \geq \frac{1}{2}$ and $C^*_j \leq 2^\tl$ are satisfied. Intuitively, $\tl(j)$ can be viewed as a rough estimate of the completion time of task $\task{j}$. Define $\alpha (j)$ as the total fraction of task $\task{j}$ over any machine in the first $\tl(j)$ intervals with respect to solution $x^*$:
\begin{equation*}
    \alpha_j = \sum_{t=1}^{\tl(j)} \sum_i x^*_{i,j,t}.
\end{equation*}

We construct a set of feasible solutions $\tilde{x}$ based on the optimal solution $x^*$ for the LP:
\begin{equation} \label{eq:add_chill}
    \tilde{x}_{i,j} = \sum_{\tl=1}^{\tl(j)} \frac{x^*_{i,j,\tl}}{\alpha_j} \quad \forall i, j.
\end{equation}
Notice that the group assignment rule $\gaw{\cdot}$ is of the same form as $\gam{\cdot}$, with $\tilde{x}$ replacing $x^*$. For  task \task{j}, define $\tilde{\ell}_j$ as before but with respect to $\tilde{x}$ instead of $x^*$:
\begin{equation*}
    \tilde{\ell}_j = \max_\ell \ell \quad \text{s.t.} \quad \sum_{k=\ell}^K  \tilde{x}_{M_k, j} \geq \frac{1}{2}.
\end{equation*}
The group assignment rule $\gaw{\cdot}$ for the goal of minimizing the total weighted completion time follows as below:
\begin{equation*}
    \gaw{j} = \argmax_{M_k :\tilde{\ell}_j \leq k \leq K} \quad s(M_k).
\end{equation*}

\subsection{A Comparison of GETF and SLS} \label{section:comparison}
% \shai{the title reads weird to me... Maybe "a comparison of...". Not a big deal... you can keep it as is if you like it}
\begin{figure}[htp]
\centering
\subfloat[]{%
  \includegraphics[width=0.45\textwidth]{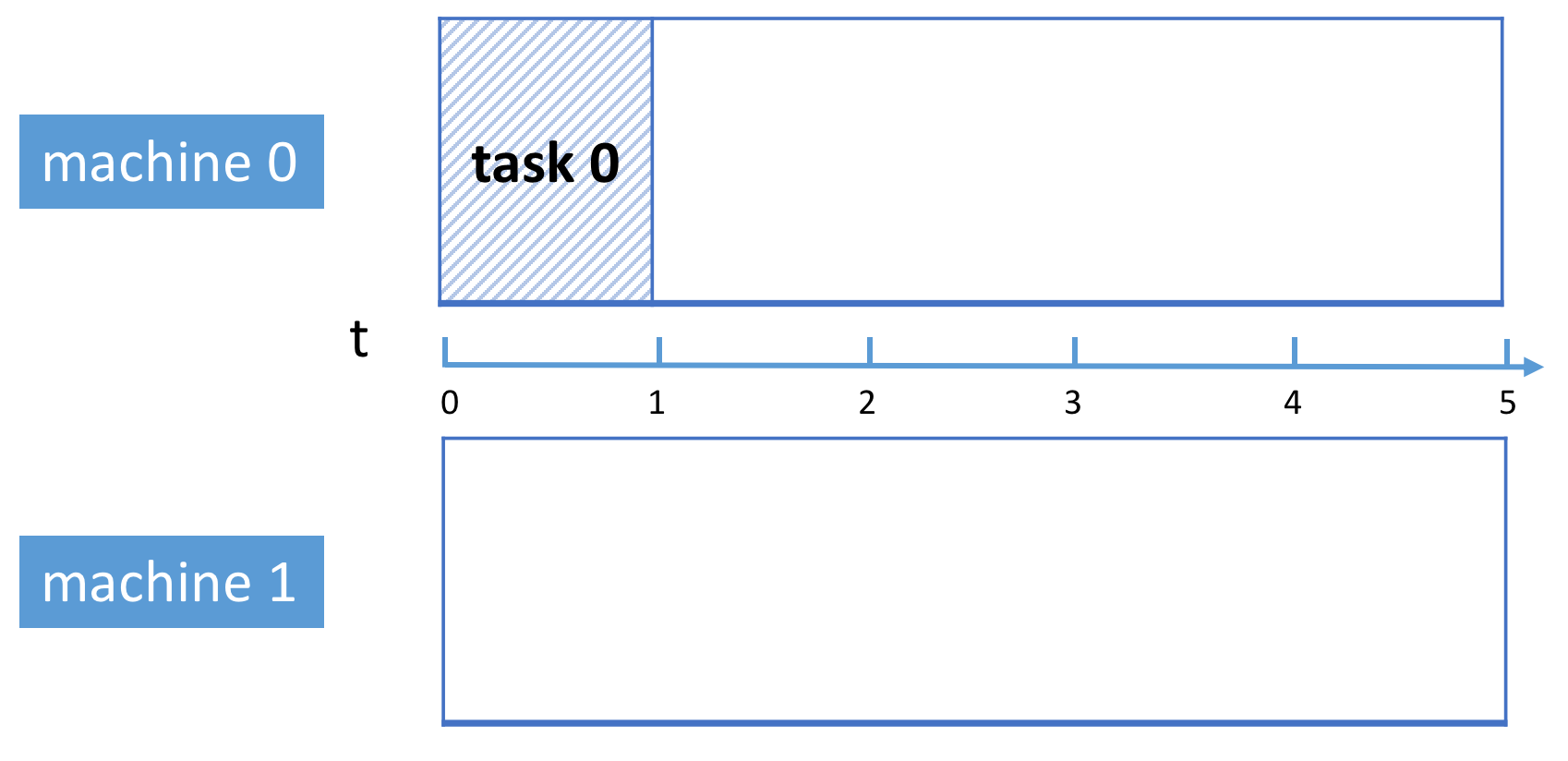}%
}%
\subfloat[]{%
  \includegraphics[width=0.45\textwidth]{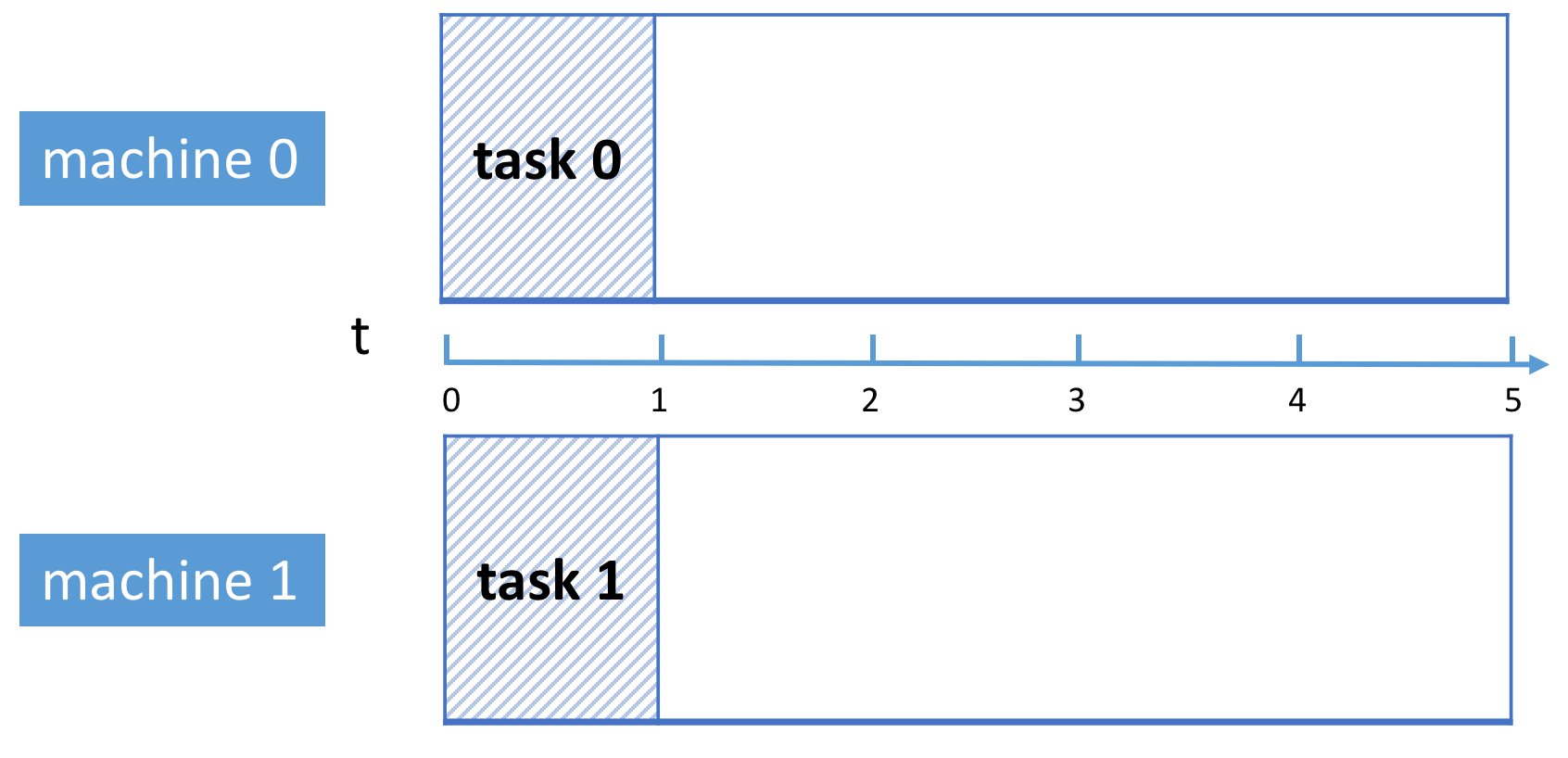}%
}

\subfloat[]{%
  \includegraphics[width=0.45\textwidth]{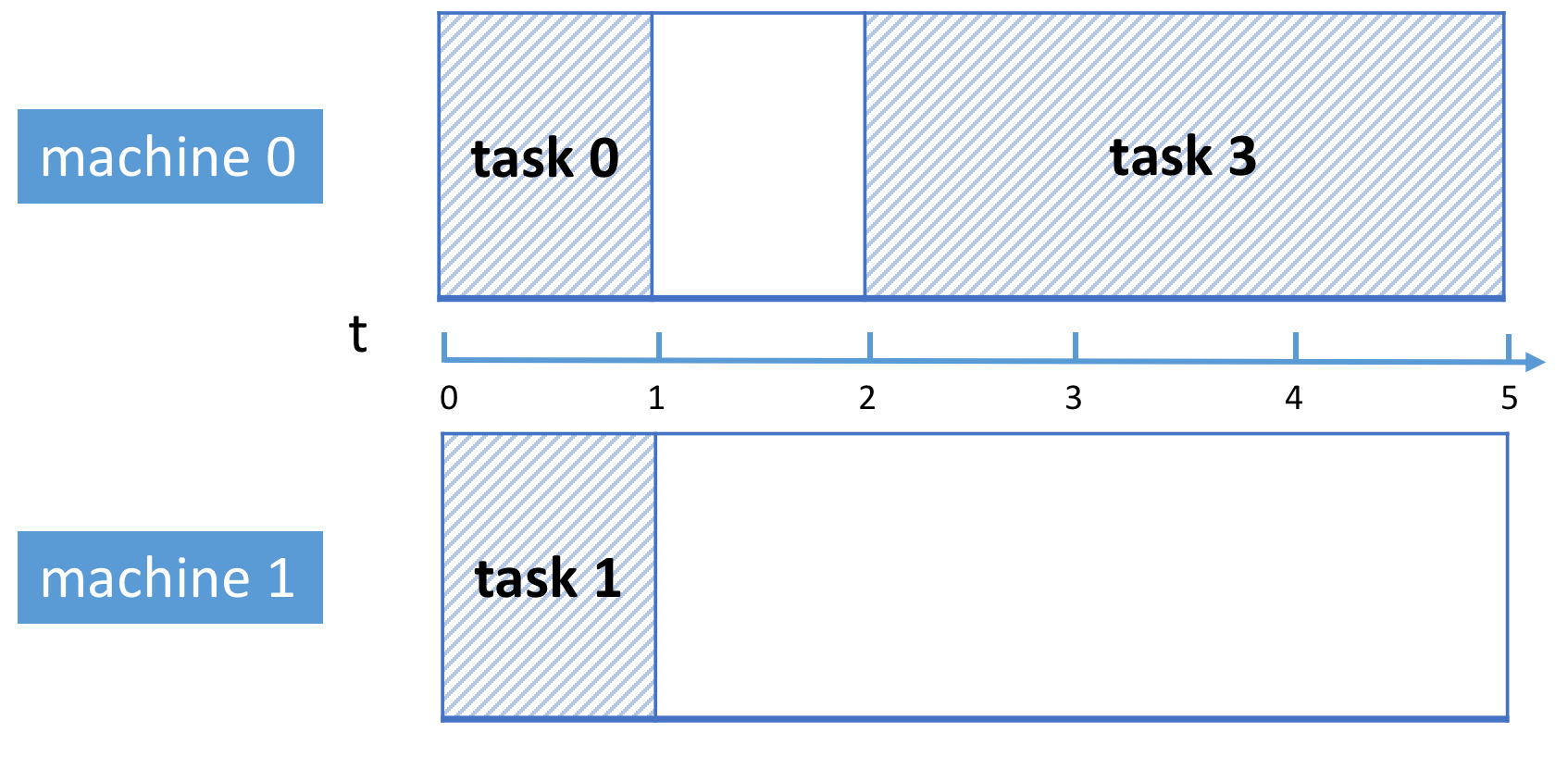}%
}%
\subfloat[]{%
  \includegraphics[width=0.45\textwidth]{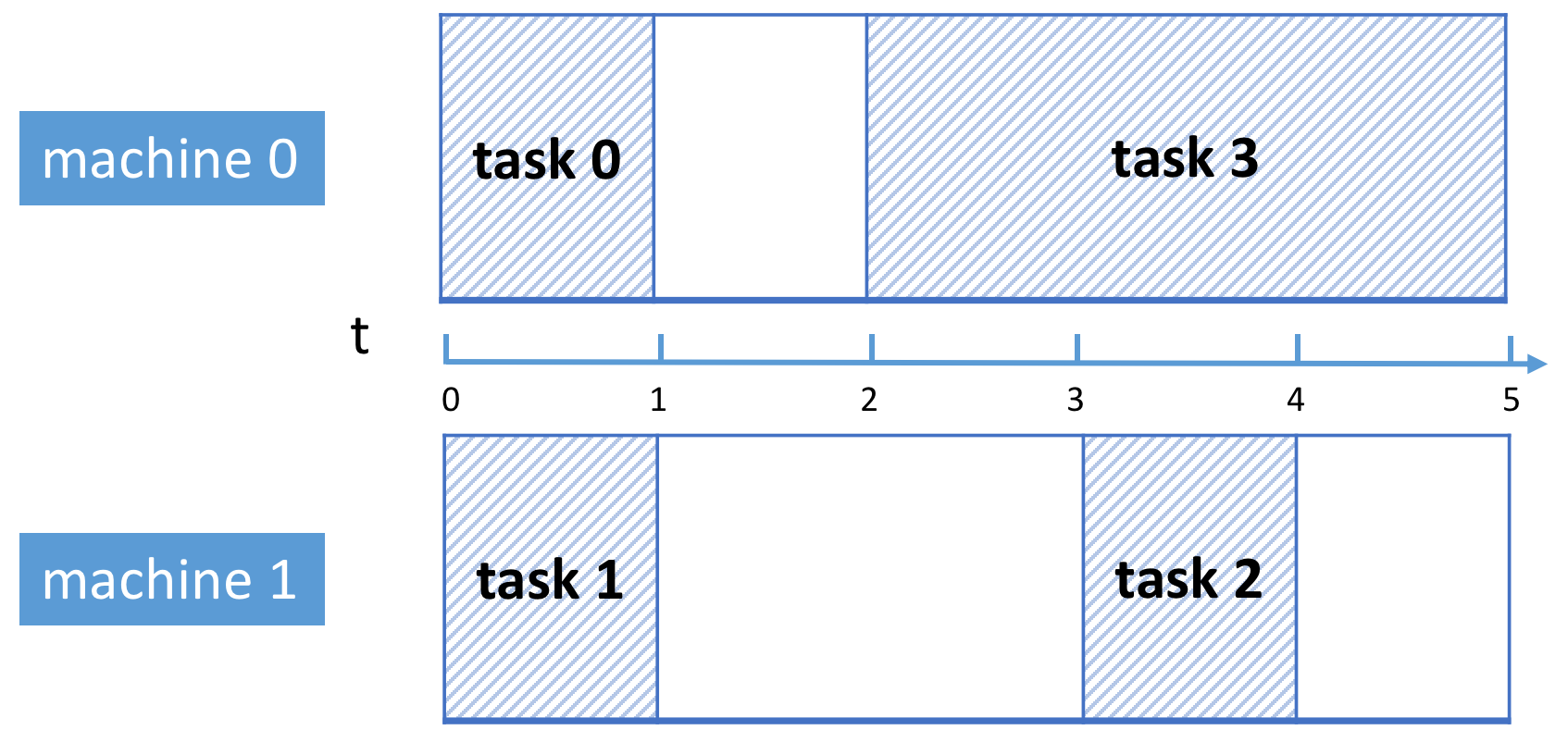}%
}

\caption{An illustration of GETF running on Example \ref{example:etf}. (a)-(d) show the first four iterations. \label{diagram:etf}}

\end{figure}

\begin{figure}[htp]
\centering
\subfloat[]{%
  \includegraphics[width=0.45\textwidth]{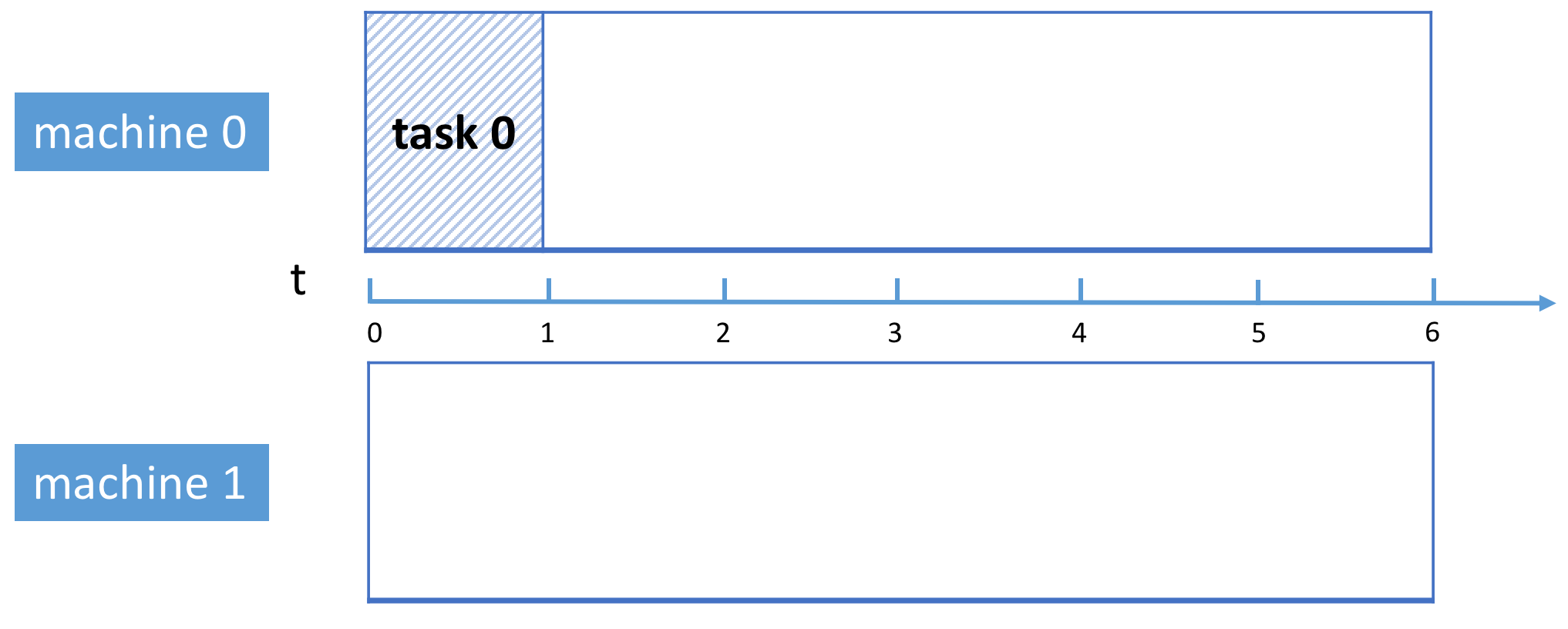}%
}%
\subfloat[]{%
  \includegraphics[width=0.45\textwidth]{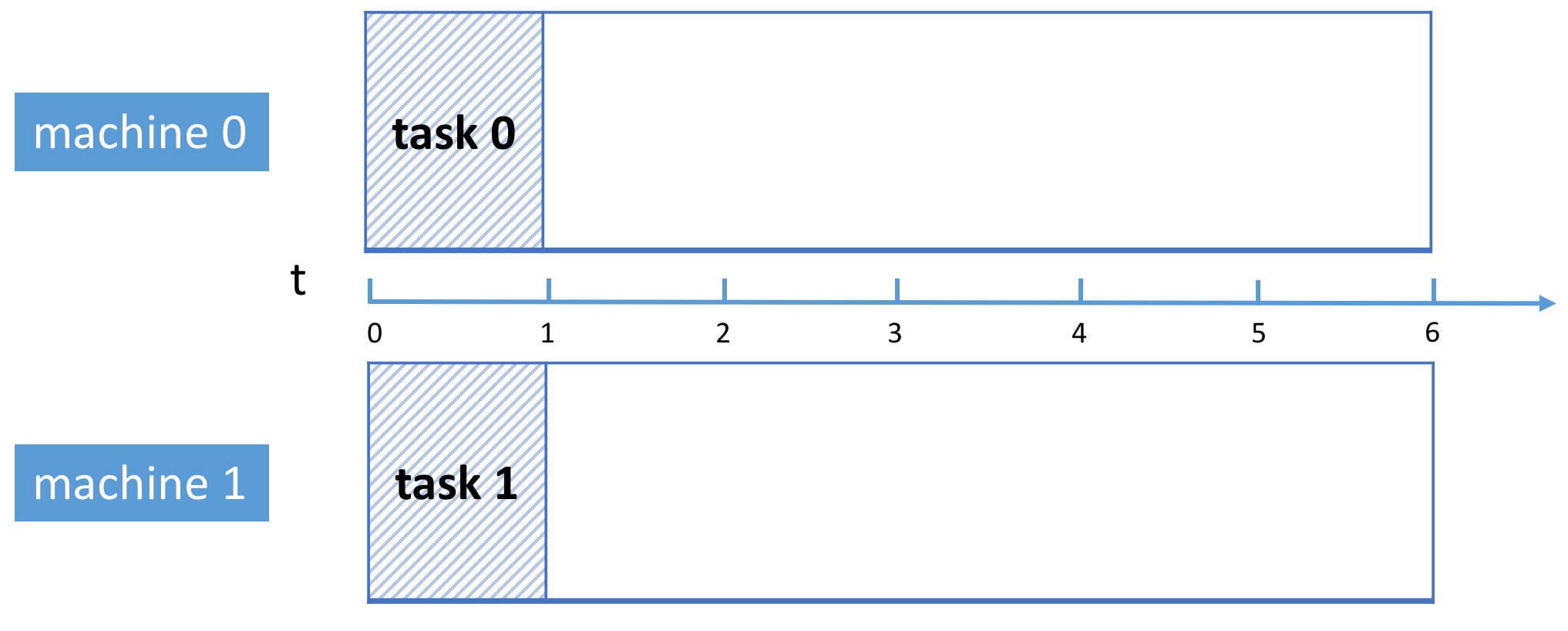}%
}

\subfloat[]{%
  \includegraphics[width=0.45\textwidth]{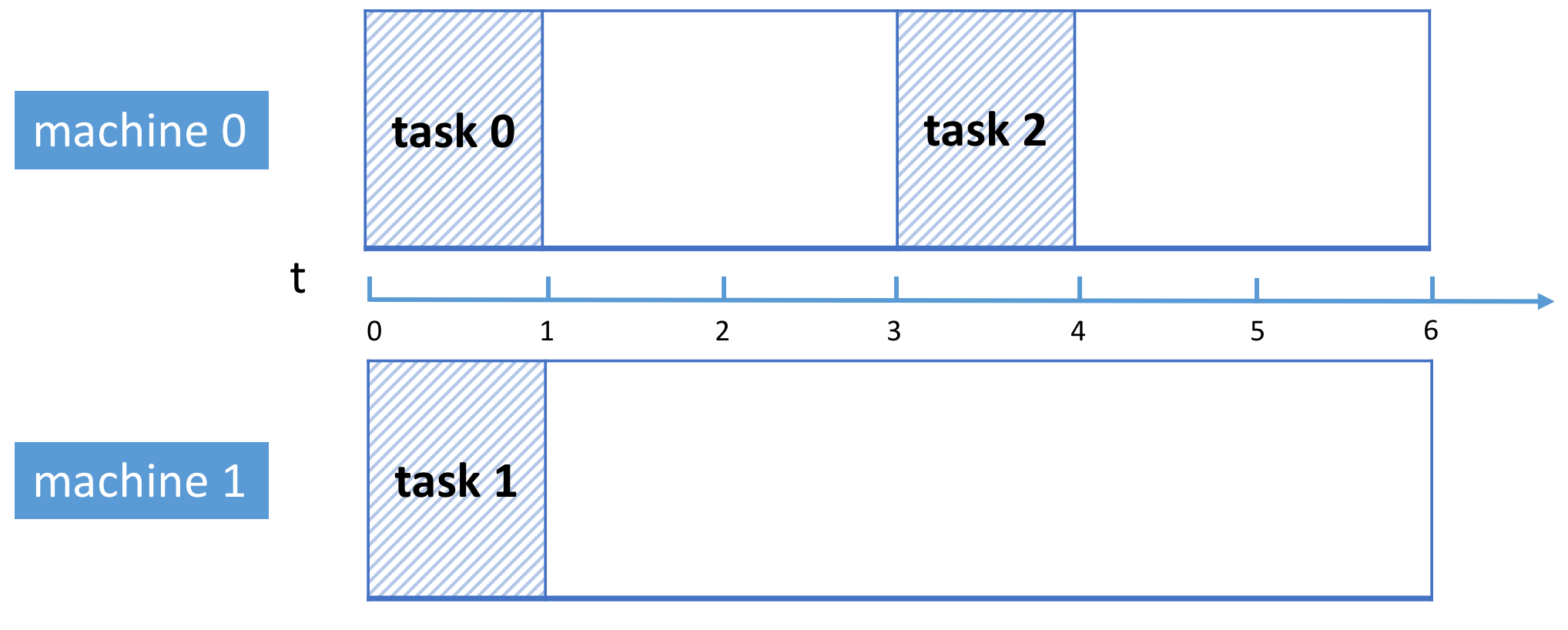}%
}%
\subfloat[]{%
  \includegraphics[width=0.45\textwidth]{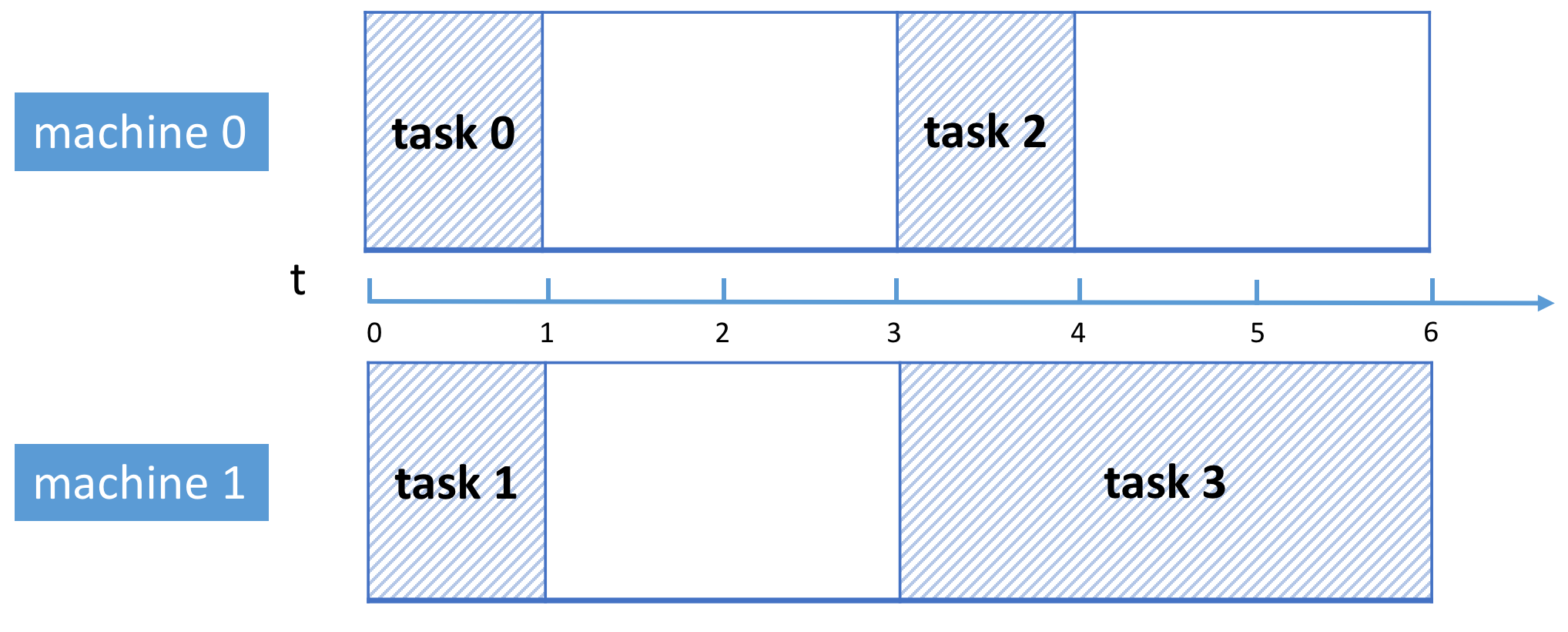}%
}

\caption{An illustration of SLS running on Example \ref{example:etf}.  (a)-(d) show the first four iterations. \label{diagram:ls}}

\end{figure}

The description of GETF above highlights that it combines the greedy heuristic of ETF with the speed-based assignment heuristic of SLS.  This enables GETF to provide guarantees for settings with both heterogeneous processing rates and communication delays. In contrast, SLS does not provide guarantees in settings with communication time.  This is a result of the fact that SLS is based on list scheduling and does not always schedule the earliest task first, thus making it impossible to bound the overall idle time in between tasks.

% \begin{figure}[htp]
%     % \centering
    
%     \begin{minipage}[t]{0.4\linewidth}
%     \centering
%         \begin{subfigure}{\linewidth}
%             \includegraphics[width=0.8\linewidth]{diagrams/etf0.png}
%             \caption{}
%         \end{subfigure}
%         % \hfill
%         \begin{subfigure}{\linewidth}
%             \includegraphics[width=0.8\linewidth]{diagrams/etf1.png}
%             \caption{}
%         \end{subfigure}
%         \begin{subfigure}{\linewidth}
%             \includegraphics[width=0.8\linewidth]{diagrams/etf2.png}
%             \caption{}
%         \end{subfigure}
%         \begin{subfigure}{\linewidth}
%             \includegraphics[width=0.8\linewidth]{diagrams/etf3.png}
%             \caption{}
%         \end{subfigure}
%     \caption{etf}
%     \end{minipage}
%     % \hspace{\stretch{2}}%
%     \begin{minipage}[t]{0.4\linewidth}
%     \centering
%         \begin{subfigure}{\linewidth}
%             \includegraphics[width=0.8\linewidth]{diagrams/ls0.png}
%             \caption{}
%         \end{subfigure}
%         % \hfill
%         \begin{subfigure}{\linewidth}
%             \includegraphics[width=0.8\linewidth]{diagrams/ls1.png}
%             \caption{}
%         \end{subfigure}
%         \begin{subfigure}{\linewidth}
%             \includegraphics[width=0.8\linewidth]{diagrams/ls2.png}
%             \caption{}
%         \end{subfigure}
%         \begin{subfigure}{\linewidth}
%             \includegraphics[width=0.8\linewidth]{diagrams/ls3.png}
%             \caption{}
%         \end{subfigure}
%     \caption{etf}
%     \end{minipage}
% \end{figure}

To illustrate the difference between GETF and SLS, we provide a simple example of scheduling a job made up of four tasks.

\begin{example} \label{example:etf}  \textit{
We consider a job made up of four tasks, $\task{0}, \task{1}, \task{2}, \task{3}$ with processing demands $1, 1, 1, \text{ and } 3$ that are to be scheduled on a set of two identical machines with the same processing speed equal to $1$. The weight for the edges in the graph are listed as below: $\edgeW{0}{2} = \edgeW{0}{3} = \edgeW{1}{2} = 2, \edgeW{1}{3} = 1$. We assume $\cs{i}{j} = 1$ for $i \neq j$; otherwise $\cs{i}{i} = 2$ for $i=0,1$.  }

\textit{The schedules of GETF and SLS are illustrated in Figures \ref{diagram:etf} and \ref{diagram:ls}. Note that, since the servers are identical, the group assignment rule does not play a role in these examples.  Given a priority list $(\task{0}, \task{1}, \task{2}, \task{3})$, a possible schedule produced by SLS puts tasks $\task{0}$ and $\task{2}$ on machine $\machine{0}$ and assigns the rest of tasks to machine $\machine{1}$ as demonstrated in Figure \ref{diagram:ls}. A terminal chain for the given schedule is task $\task{1}$ followed by task $\task{3}$, and the idle time of length $2$ between the end of task $\task{1}$ and the start of task $\task{3}$ on machine $\machine{1}$ is not bounded by the communication time between task $\task{1}$ and $\task{3}$. In contrast, task $\task{3}$  starts earlier on machine $\machine{0}$ in a schedule produced by GETF, see Figure \ref{diagram:etf}. List scheduling does not always schedule the earliest task at each step, thus making the idle time on machine $\machine{1}$ not necessarily bounded by communication time between task $\task{1}$ and task $\task{3}$.  Our proofs in Section \ref{section:main} highlight that maintaining a tight bound on the communication time between tasks is crucial to achieving a good approximation ratio in settings with machine-dependent communication time.}
\end{example}

\section{Results}
Our main results bound the approximation ratio of GETF in settings with related machines and heterogeneous communication time for  the goals of minimizing the makespan and minimizing the total weighted completion time. %\sout{This is the first algorithmic framework with a provable approximation ratio in this setting for either performance measure.}\shai{if we  had not impressed this upon the reader by the end of the intro, now is too late. If we had, this is unnecessarily repetitious. If you feel it wasn't stressed enough in the intro, stress it more in the intro}

\subsection{Makespan} \label{section:main}
In the case of minimizing the makespan, our main result provides a bound in terms of the communication time of a terminal chain of the schedule. Specifically, let $\tc: \task{c_1} \prec \task{c_2} \prec \ldots \prec \task{c_N}$ be a terminal chain for the schedule and define $C$ as the communication time over such a chain in the worst case, i.e. 
\begin{equation*}
C = \sum_{j=2}^{N}  \frac{\edgeW{c_{j-1}}{c_{j}}}{ \sbar{j-1}{j}},
\end{equation*}
where $\sbar{j-1}{j}$ is defined as the slowest speed between $\ma{c_{j-1}}$,  the machine assigned to $\task{c_{j-1}}$ and any machine in the group ${\ga{c_{j}}}$,  i.e.,
\begin{equation*}
\sbar{j-1}{j} = \min_{i \in {\ga{c_{j}}}} \cs{\ma{c_{j-1}}}{i}.
\end{equation*}
Note that $C$ can be computed efficiently and minimized over all the terminal chains using dynamic programming and that the tie-breaking rule can have an impact on $C$ due to its impact on terminal chains.

\begin{theorem}\label{theorem:main}
For any schedule $\mathcal{S}$ produced by GETF with group assignment rule, $\gam{\cdot}$ 
\begin{equation*}
    C_{max} (\mathcal{S}) \leq \bigO(\loglogm) \opti + C,
\end{equation*}
where $\opti$ is the optimal schedule length obtained if communication time for all pairs were zero.
\end{theorem}

Theorem \ref{theorem:main} represents the first result for makespan in the setting of related machines and heterogeneous communication time, addressing a problem that has been open since ETF was introduced for identical machines thirty years ago. Additionally, it matches the state-of-the art results for the case without communication time, where the best known approximation ratio is $\bigO(\loglogm)$ \cite{li2017scheduling}, and the case with communication time but identical machines, where the best known approximation ratio is $(2 -\frac{1}{m}) \opti + C$ \cite{hwang1989scheduling}.

Concretely, in the special case of identical machines, the group assignment rule $\gam{\cdot}$ is no longer required when implementing GETF since all machines share the same speed and so there is only one group of machines.  Thus, GETF reduces to ETF.  The theorem makes use of $C'$ which is defined as
$$C' = \frac{1}{m} \sum_{j=2}^{N} \sum_{i=1}^m \frac{\edgeW{c_{j-1}}{c_{j}}}{s_{\ma{c_{j-1}}, i}}.$$ Note that $C'$ differs from $C$ since it is an average over the terminal chain.  The result we obtain in this case is the following, which matches the current state-of-the-art result of \cite{hwang1989scheduling}.

\begin{proposition}\label{prop:identical}
Consider a setting with $m$ identical machines.  For any schedule $\mathcal{S}$ produced by GETF,
\begin{equation*}
    C_{max} (\mathcal{S}) \leq \left(2 -\frac{1}{m}\right) \opti + C', 
\end{equation*}
where $\opti$ is the optimal schedule length obtained if communication time for all pairs were zero.
\end{proposition}

\subsection{Total Weighted Completion Time} \label{section:weighted}

Similarly to the  makespan case, we provide a bound with respect to the communication time of chains. However, since total weighted completion time depends on the completion time of every task (instead of just one task as in the case of makespan), the communication time of terminal chains of many subsets of the DAG show up in the bound. More formally, assume that the tasks are indexed with respect to their order in the schedule determined by GETF, denoted by $\mathcal{S}$. At iteration $j$, task $\task{j}$ is to be scheduled.  Let $G(\mathcal{S}, j)$ denote a DAG formed by a set of the tasks that have been scheduled so far and the corresponding edges within these tasks. Define $\mathcal{S}(j)$ to be a subset of the given schedule $\mathcal{S}$ up to iteration $j$, i.e., it is a schedule for DAG $G(\mathcal{S}, j)$. This definition ensures that task $\task{j}$ is one of the tasks that ends last in the schedule $\mathcal{S}(j)$.  Now, let $\tc (\mathcal{S}, j): \task{c_1} \prec \task{c_2} \prec \cdots \prec \task{c_{N_j}}$ be a terminal chain that ends with task $\task{j} = \task{c_{N_j}}$ in the schedule $\mathcal{S}(j)$, and define $C(\mathcal{S}, j)$ 
% \note{(this notation is chosen to be consistent with notation used in the literature, e.g. \cite{hwang1989scheduling})} \adam{I don't think this phrase is needed}
as the communication time over such a chain in the worst case, i.e.,
\begin{equation*}
C(\mathcal{S}, j) = \sum_{j'=2}^{N_j}  \frac{\edgeW{c_{j'-1}}{c_{j'}}}{ \sbar{j'-1}{j'}}.
\end{equation*}
This definition of $C(\mathcal{S}, j)$ generalizes the notion of $C$ used in Theorem \ref{theorem:main} for makespan and plays a similar role in the theorem below. 

\begin{theorem} \label{theorem:weighted_main}
For any schedule $\mathcal{S}$ produced by GETF with group assignment rule $\gaw{\cdot}$,
\begin{equation*}
    \sum_j \weight{j} C_j \leq \bigO(\loglogm) \wopti + \sum_j \omega_j C(\mathcal{S}, j),
\end{equation*}
where $\wopti$ is the optimal total weighted completion time obtained if communication time for all pairs was zero.
\end{theorem}

Theorem \ref{theorem:weighted_main} is the first result on total weighted completion time for the setting of related machines with heterogeneous communication time and it matches the bounds in cases where previous results exist. In particular, if the weights are chosen so as to recover makespan, then the bound matches that of Theorem \ref{theorem:main}.  Similarly, results for identical machines can be recovered as done in the case of makespan.   However, note that the group assignment rule used for GETF here is different than that in Theorem \ref{theorem:main}.  The rule used in Theorem \ref{theorem:weighted_main} applies more generally but, while both group assignment rules yield the same worst-case performance bound for makespan, we expect that the rule used in Theorem \ref{theorem:main} will lead to a smaller makespan in most practical settings as it is designed for the purpose of minimizing the makespan.

\section{Proofs} 
In this section, we present our proofs of Theorems \ref{theorem:main} and \ref{theorem:weighted_main}.  The general form of both arguments is similar; however, the case of total weighted completion time is more involved.  The first step of our argument is to show a general upper bound, which is valid for  GETF regardless of choices of group assignment function $\ga{\cdot}$, and tie-breaking rule. This \textit{Separation Principle} can be used to easily establish the result for makespan in the case of identical machines (Proposition \ref{prop:identical}), and represents a significant simplification compared to existing proofs of that result in the literature.  We then tighten the general bound by taking advantage of the choices of $\ga{\cdot}$ described in Section \ref{section:GETF} for makespan and total weighted completion time. Finally, we establish a connection between the makespan and total weighted completion time in the same settings by introducing a time-indexed LP that enables us to bound the total weighted completion time.

\subsection{A Separation Principle}

The Separation Principle presented here is a key component of our proof of Theorem \ref{theorem:main}. The core of nearly all proofs in this area is the construction of a chain, which is then used to bound the overall makespan. This idea goes back to the first list scheduling algorithms proposed by \cite{graham1969bounds}. The key to our  argument is to bound the amount of communication time between any predecessor-successor pairs in a terminal chain.  However, as we discuss in Section \ref{section:GETF}, it is not possible to do this under list scheduling algorithms.%, and so our approach differs from that in \cite{li2017scheduling}.}   

Our approach also differs considerably from the approach used to study ETF in \cite{hwang1989scheduling}, where the authors divide $[0, C_{max}]$ into two sets of time intervals, one for the time when all the machines are busy and the other that one chain covers. Extending this approach to related machines does not appear possible. In contrast, in our argument, the construction of a terminal chain is simple and so we can identify the set of time intervals between tasks in the terminal chain and take advantage of the greedy nature of GETF to bound these times directly.  

A key feature of the the Separation Principle below is that it separates the analysis of the terminal chain from the analysis of the group assignment rule, which provides another valuable simplification of the previous proof approaches.

\begin{theorem}[Separation Principle]
For any choice of group assignment function $\ga{\cdot}$ and tie-breaking rule, GETF produces a schedule $\mathcal{S}$ of makespan
\begin{equation*}
    C_{max}(\mathcal{S}) \leq P + \sum_{k=1}^K D_k + C,
\end{equation*}
where
\begin{equation*}
    P = \sum_{\task{c_j} \in \tc} \frac{\nodeW{{c_j}}}{\ps{\ma{c_j}}},
\end{equation*}
\begin{equation*}
    D_k =  \frac{\sum_{\task{j}: k \in \ga{j} } \nodeW{j}}{s(M_k)},
\end{equation*}
\begin{equation*}
    C = \sum_{j=1}^{N-1}  \frac{\edgeW{c_{j}}{c_{j+1}}}{ \sbar{j}{j+1}}.
\end{equation*}
\end{theorem}

Note that the upper bound in this result is valid regardless of the choice of group assignment rule and tie-breaking rule. $P$ is the sum of processing times along a terminal chain and $D_k$ can be viewed as total load assigned to machines in group $M_k$. Both $P$ and $D_k, k= 1, 2, \ldots, K$, are not dependent on the communication constraint, which enables us to take advantage of any good choice of group assignment rule $\ga{\cdot}$ for general DAG scheduling, even in the case of zero communication time.

\begin{proof}

 Our proof proceeds in four steps:
\begin{enumerate}[label=(\roman*)]
\item Define a terminal chain $\tc$. Recall that a chain $\tc$, $\task{c_1}  \prec \task{c_2} \prec \ldots \prec \task{c_N}$ is a terminal chain when task $\task{c_N}$ completes at the end of the overall schedule.
\item Partition the overall makespan into $K+1$ parts.  The idea of this step is to decouple $[0, C_{max}]$ into one part where the tasks in the terminal chain are being processed and $K$ other parts associated with each machine group. Dependent on the choices of group assignment rule, we can further bound these $K+1$ parts. 
\item Bound the idle time in between tasks. The greedy nature of GETF makes it possible to bound the length of the idle time intervals between tasks by communication delays of task pairs.
\item Combine (ii) and (iii) to bound the overall makespan in terms of the communication time of the terminal chain.
\end{enumerate}

\textit{$(\romannum{1})$ Define a terminal chain $\tc$.} To find a terminal chain of length $N$, we start with one of the tasks that ends last, denoted as $\task{c_N}$. According to the definition of $\ma{\cdot}$ and $\st{\cdot}$, task $\task{c_N}$ is assigned to machine $\ma{c_N}$ in group ${\ga{c_N}}$ with a starting time $\st{c_N}$. Among all the immediate predecessors of task $\task{c_N}$, we pick one of the tasks that finishes last and define it as $\task{c_{N-1}}$. In such a fashion, we construct a chain $\tc$ of tasks $\task{c_1} \prec \task{c_2} \prec \ldots \prec \task{c_N}$ of length $N$ such that $\task{c_1}$ does not have any predecessor.

\textit{$(\romannum{2})$ Partition $[0, C_{max}]$ into $K+1$ parts, $\mathcal{T}_0, \mathcal{T}_1, \ldots, \mathcal{T}_K$.} Recall that $K=\bigO{(\loglogm)}$ is the number of groups for machines by the group assignment rule as we describe in the previous section. Let $\mathcal{T}_0$ denote the union of the time intervals during which tasks of chain $\tc$ are being processed. Consider the time interval between the end of task $\task{c_{j-1}}$ and the start of task $\task{c_{j}}$ for $j = 2, 3, \ldots, N$, and assign it to $\mathcal{T}_{k}$ where $M_k = \ga{c_{j}}$. As a set of time intervals, $\mathcal{T}_k$ can be possibly empty or have more than one time interval. Essentially, $\mathcal{T}_k$ is a set of time intervals that tasks in the terminal chain $\tc$ assigned to machines in group $M_k$ have to wait before being processed. In such a fashion, we define $\mathcal{T}_1, \mathcal{T}_2, \ldots, \mathcal{T}_K$ since $\ga{\cdot}$ maps each task to one of the $K$ machine groups. The length of the union of $\mathcal{T}_i$ for $i=0,1,\ldots,K$ is the makespan.

\textit{$(\romannum{3})$ Bound the idle time in between tasks.} Consider a task $\task{c_j}$  assigned to machine $\ma{c_j}$. For each machine $\machine{i} \in {\ga{{c_j}}}$, let $E(c_{j-1}, c_j, i)$ denote a union of disjoint empty time intervals on machine $\machine{i}$ between the end time of task $\task{c_{j-1}}$ and the start time of task $\task{c_j}$. Between the end time of task $\task{c_{j-1}}$ and the start time of task $\task{c_j}$, there can be multiple tasks being processed on machine $\machine{i}$ in serial, possibly resulting in  more than one idle time interval on machine $\machine{i}$ during that time interval $E(c_{j-1}, c_j, i)$. Precedence constraints between task pairs can also possibly make a successor wait before it gets started. Regardless of the reason for idle time between tasks, each task can not possibly start earlier on any machine in the assigned group due to the greedy feature of GETF. Thus the length of $E(c_{j-1}, c_j, i)$ is bounded above by the communication time between task $\task{c_{j-1}}$ and task $\task{c_j}$, i.e.,
\begin{equation*}
    |E(c_{j-1},c_j, i)| \leq \frac{\edgeW{c_{j-1}}{c_{j}}}{\cs{\ma{c_{j-1}}}{i}} \quad \forall \machine{i} \in {\ga{c_j}}.
\end{equation*}
This is true because if it were not the case then task $\task{c_j}$ could have started earlier on machine $\machine{i}$. Note that the end time of task $\task{c_j}$ could possibly be earlier if it were allowed to be scheduled on a faster machine with a slightly bigger communication delay, since the processing speeds of machines in the same group vary.

Let $\e{i}$ be idle time on machine $\machine{i}$ in group $M_k$ during the time interval $\mathcal{T}_k$, and let $\ebar{k}$ be maximum idle time on any machine in group $M_k$ during the time intervals $\mathcal{T}_k$, i.e., $\e{i} \leq \ebar{k}$ for all $\machine{i} \in M_k$. Thus,

\begin{alignat}{2}
\sum_{k=1}^{K}  \ebar{k} & \leq  \sum_{j=2}^{N}  \frac{\edgeW{c_{j-1}}{c_{j}}}{\min_{\machine{i'} \in {\ga{c_j}}} \cs{\ma{c_{j-1}}}{i'}} \notag\\
& \leq \sum_{j=2}^{N}  \frac{\edgeW{c_{j-1}}{c_{j}}}{ \sbar{j-1}{j}}. \label{eq:ebar}
\end{alignat}

\textit{$(\romannum{4})$ Bound the makespan.} For $1 \leq k \leq K$, the total speed of machines in group $M_k$ is $$s(M_k) = \sum_{\machine{i} \in M_k} \ps{i}.$$
 Denote the total length of the intervals in $\mathcal{T}_k$ by $t_k$. There must be at least a sum of $\left(t_k - \e{i}\right) \ps{i}$ units of processing done on each machine $\machine{i}$ in group $M_k$ during the time intervals $\mathcal{T}_{k}$. Thus for $1 \leq k \leq K$,
\begin{equation*}
    \sum_{\machine{i} \in M_k} \left(t_k - \e{i} \right) \ps{i} \leq \sum_{\task{j}: \ga{j} = M_k} \nodeW{j}.
\end{equation*}
Therefore,
\begin{equation}
     t_k \leq \frac{\sum_{\task{j}: \ga{j} = M_k} \nodeW{j}}{s({M_k})} + \frac{\sum_{\machine{i} \in M_k} \e{i} \ps{i}}{s(M_k)}. \label{eq:tea}
\end{equation}
We now bound $\mathcal{C}_{max}$: 
\begin{subequations}
\begin{alignat}{3}
\mathcal{C}_{max}   =& \quad \sum_{k=1}^K t_k + t_0\notag \\
 \leq &\quad \sum_{k=1}^K \left( \frac{\sum_{\task{j}: \ga{j} = k} \nodeW{j}}{s({M_k})} + \frac{\sum_{\machine{i} \in M_k} \e{i} \ps{i}}{s(M_k)} \right) + \notag \\
& \quad \sum_{\task{c_j} \in \tc} \frac{ \nodeW{c_j}}{\ps{\ma{c_j}}} \label{eq:chai}\\
 \leq& \quad P + \sum_{k=1}^{K} D_k + \sum_{k=1}^K \ebar{k} \frac{\sum_{\machine{i} \in M_k} \ps{i}}{s(M_k)} \notag\\
 = &\quad P + \sum_{k=1}^{K} D_k + \sum_{k=1}^K \ebar{k}  \notag\\
 \leq &\quad P + \sum_{k=1}^{K} D_k + C \label{eq:yeah},
\end{alignat}
\end{subequations}
where \eqref{eq:chai} is due to \eqref{eq:tea} and \eqref{eq:yeah} is due to \eqref{eq:ebar}.
% \Halmos
\end{proof}

\subsection{Proof of Theorem \ref{theorem:main}} \label{section:main_proof}

In order to apply the Separation Principle to prove Theorem \ref{theorem:main}, we need to prove bounds on $P$ and $\sum_{k=1}^K D_k$ in the case of the group assignment rule defined in Section \ref{section:GETF}.  For this, we consider the scheduling problem with zero communication time. Note that the design of group assignment function $\gam{\cdot}$ is based on the optimal solution $x^*$ of the relaxed LP for a scheduling problem with zero communication time, hence the upper bounds for both $P$ and $\sum_{k=1}^K D_k$ are associated with the optimal objective of the relaxed LP in the setting with zero communication time as well.

The bounds of $P$ and $\sum_{k=1}^K D_k$ are given in the following two lemmas, which are adapted from results in \cite{li2017scheduling}. Theorem \ref{theorem:main} follows directly from these two lemmas, the Separation Principle, and the fact that  $T^* \leq \opti$, where $T^*$ is the optimal solution to the LP.

\begin{lemma}\label{lemma:p}
$P \leq 2 \gamma T^*.$ 
\end{lemma}
\begin{proof}
Recall that $x_{M', j}^* = \sum_{i \in M'} x^*_{i, j}$ and  $\ell_j$ as the largest group index such that at least more than half of tasks are assigned to machines in groups $M_\ell, \ldots, M_K$. For every task $\task{j}$ and any machine $\machine{i} \in {\ga{j}}$, by definition of the largest index $\ell_j$,
\begin{equation} \label{eq:ss}
    \sum_{k=1} ^{\ell_j} x^*_{M_k, j} > \frac{1}{2}.
\end{equation}
Thus,
\begin{subequations} \label{eq:lan}
\begin{alignat}{2}
\sum_{\machine{i'} \in M} \frac{x^*_{i', j}}{\ps{i'}} & = \sum_{k=1}^{K}  \sum_{\machine{i'} \in  M_k} \frac{x^*_{i', j}}{\ps{i'}} \\
& \geq \sum_{k=1}^{\ell_j}  \sum_{\machine{i'} \in  M_k} \frac{x^*_{i', j}}{\ps{i'}} \notag \\
& \geq \frac{1}{2} \gamma^{-\ell_j} \label{eq:s2}\\
& \geq \frac{1}{2\gamma \ps{i}} \label{eq:s3},
\end{alignat}
\end{subequations}
where \eqref{eq:s2} is due to \eqref{eq:ss} and the fact that processing speed of machine $i'$ in group $M_k$ for task $\task{j}$ is at most $\gamma^{\ell_j}$ for $k \leq \ell_j$, and \eqref{eq:s3} is due to the fact that processing speed of machine $i$ in group ${\ga{j}}$, whose group index is not smaller than $\ell_j$, is at least $\gamma^{\ell_j-1}$. Using this, we can bound $P$ as follows:
\begin{subequations}
\begin{alignat}{2}
    P & = \sum_{\task{c_j} \in \tc} \frac{\nodeW{c_j}}{\ps{\ma{c_j}}} \notag \\
    & \leq 2 \gamma \sum_{\task{c_j} \in \tc} \nodeW{c_j} \sum_{{i'} \in M} \frac{x^*_{i', c_j}}{\ps{i'}} \label{eq:f1}\\
    & \leq 2 \gamma \sum_{\task{c_j} \in \tc} C^*_{c_j} \label{eq:f2}\\
    & \leq 2 \gamma T^* \label{eq:f3},
\end{alignat}
\end{subequations}
where \eqref{eq:f1} is due to \eqref{eq:lan}, \eqref{eq:f2} is due to constraint \eqref{cons:cap_each_machine} of the LP and \eqref{eq:f3} is due to constraint \eqref{cons:prec} of the LP.
% \Halmos
\end{proof}

\begin{lemma}\label{lemma:D_k}
$\sum_{k=1}^K D_k \leq 2KT^*.$
\end{lemma}
\begin{proof}
For any task $\task{j}$, by definition of $\ell_j$, $\sum_{k=\ell_j}^{K} x^*_{M_k, j} \geq \frac{1}{2}$. Thus,
\begin{alignat}{2}
    \frac{1}{2 s({\ga{j}})} & \leq \sum_{k=\ell_j}^K \frac{x^*_{M_k, j}}{s({\ga{j}})} \notag \\
    & \leq \sum_{k=\ell_j}^K \frac{x^*_{M_k, j}}{s(M_k)} \label{eq:shan} \\
    & \leq \sum_{k=1}^K \frac{x^*_{M_k, j}}{s(M_k)} \notag.
\end{alignat}
Inequality \eqref{eq:shan} is due to the fact that the assigned group $f(j)$ maximizes the total speeds of machines in that group among the candidates $M_{\ell_j}, \ldots, M_K$. Thus,
\begin{alignat}{2}
\sum_{k=1}^K D_k & = \sum_{k=1}^K \frac{\sum_{\task{j}: \ga{j} = M_k} \nodeW{j}}{s(M_k)}  = \sum_{\task{j} \in V} \frac{\nodeW{j}}{s({\ga{j}})} \notag \\
& \leq 2 \sum_{\task{j} \in V} \nodeW{j} \sum_{k=1}^K \frac{x^*_{M_k, j}}{s(M_k)} \notag\\
& = 2 \sum_{k=1}^K \frac{1}{s(M_k)} \sum_{\task{j} \in V} \nodeW{j} x^*_{M_k, j} \notag \\
& \leq 2 \sum_{k=1}^K T^*  \label{eq:last} \\
& = 2KT^*. \notag
\end{alignat}
The total load assigned to machines in group $M_k$ is $\sum_{\task{j} \in V} \nodeW{j} x^*_{M_k, j}$ while its total speed is $S(M_k)$. Summing over machines in group $M_k$ on both sides for constraint \eqref{cons:cap_each_machine} leads to  \eqref{eq:last}.
% \Halmos
\end{proof}

\subsection{Proof of Proposition \ref{prop:identical}}

We now show how the Separation Principle can be used to provide a new, simpler proof of the state-of-the-art approximation ratio of ETF in the case of identical machines.  Recall that the group assignment function is not required for GETF in this case. 

To prove Proposition \ref{prop:identical}, we use the same approach as we used for proving the Separation Principle. However, we can tighten the analysis in the final step of the argument.  Specifically, the proof can be broken into three steps, instead of four:
\begin{enumerate}[label=(\roman*)]
\item Define a terminal chain $\tc$. This step is identical to the definition of a terminal chain in the proof of the Separation Principle.
\item Bound the idle time in between tasks. As the machines are identical in terms of processing speed, communication speed between different machine pairs are still heterogeneous due to the possible geolocations of machines.
\item Combine (i) and (ii) to bound the overall makespan in terms of the communication time of the terminal chain.
\end{enumerate}
Compared with the proof of the Separation Principle, Step (i) defines a terminal chain in the exactly same way. In Step (ii), bounding the idle time in the case of identical machines is also similar. Step (iii) requires more work.  Here, we further tighten the bound by eliminating the processing time of the terminal chain to improve the constant factor.

\textit{$(\romannum{1})$ Define a terminal chain $\tc$.} This step is identical to the definition of a terminal chain in the proof of the Separation Principle.

\textit{$(\romannum{2})$ Bound the idle time in between tasks.} Let $I(c_{j-1}, c_j)$ be the time interval between the end time of task $\task{c_{j-1}}$ and the start time of $\task{c_i}$ for $j=2, 3, \ldots,N$. As we explained in the Separation Principle, there can possibly be multiple idle time intervals on a machine during the time interval $I(c_{j-1}, c_j)$. For each machine $\machine{i} \in M$, define $E(c_{j-1}, c_j, i)$ as a union of disjoint empty time intervals on machine $\machine{i}$ during the time interval ${I}(c_{j-1}, c_j)$. For any machine $\machine{i}$, the length of $E(c_{j-1}, c_j, i)$ is bounded above by the communication time between task $\task{c_{j-1}}$ and task $\task{c_j}$, i.e.,
\begin{equation*}
    |E(c_{j-1},c_j, i)| \leq \frac{\edgeW{c_{j-1}}{c_{j}}}{\cs{\ma{c_{j-1}}}{i}} \quad \forall \machine{i} \in M, j = 2, 3, \ldots, N.
\end{equation*}
Otherwise task $\task{c_j}$ could have started earlier on machine $\machine{i}$.

\textit{$(\romannum{3})$ Bound the makespan.} During the time intervals ${I}(c_{j-1}, c_j)$ for $j=2,3, \ldots, N$, there must be at least $\sum_{j=2}^{N} \sum_{i=1}^m (|{I}(c_{j-1}, j_i)| - |{E}(c_{j-1}, c_j,i)|)$ processing units done, and it is bounded by a sum of the processing units for all the tasks except those in the terminal chain. This leads to the following bound:
\begin{subequations}
\begin{alignat}{2}
    \sum_{j=2}^{N} \sum_{i=1}^m \left(|{I}(c_{j-1}, c_j)| - |{E}(c_{j-1}, c_j,i)| \right) &\leq \sum_{j=1}^{n} \nodeW{j} - \sum_{j=1}^{N} \nodeW{c_j}. \label{eq:tight}
\end{alignat}
\end{subequations}
Finally, applying \eqref{eq:tight}, we have
\begin{subequations}
\begin{alignat}{4}
\mathcal{C}_{max} &=&& \sum_{j=2}^{N} |{I}(c_{j-1}, c_j)| + \sum_{j=1}^N \nodeW{c_j} \notag \\
& \leq && \frac{1}{m} \sum_{j=1}^n \nodeW{j} + \frac{m-1}{m} \sum_{j=1}^N \nodeW{c_j} + \notag \\
& && \frac{1}{m} \sum_{j=2}^{N} \sum_{i=1}^m |{E}(c_{j-1}, c_j,i)| \notag \\
& \leq && \left( 2- \frac{1}{m} \right)  \text{opt}^{(i)} + C' \label{eq:aiqing}.
\end{alignat}
\end{subequations}
The total processing time $\sum_{j=1}^n \nodeW{j}$ divided by the number of machines $m$ is the smallest possible makespan, i.e., $\frac{1}{m} \sum_{j=1}^n \nodeW{j} \leq \opti$. At the same time, the makespan of any schedule should at least cover the processing time of any chain $\tc$ in the DAG. These two facts lead to the last inequality \eqref{eq:aiqing}.

\subsection{Proof of Theorem \ref{theorem:weighted_main}}
%\shai{This section is grammatically much worse than the others. Maybe Yu Su can go over it then Adam and then let me know and I will go over it again?}
To establish the bound on the total weighted completion time for the group assignment rule $\gaw{\cdot}$, we first apply the Separation Principle to separate the requirements on communication and processing times. Second, we break the tasks into subsets based on the task completion times and, for each subset, we form an LP for those tasks alone. For each such LP, we construct a feasible solution $\tilde{x}, \tilde{C}$ and $\tilde{T}$ to bound processing time of the tasks. The feasibility of $\tilde{x}, \tilde{C}$ and $\tilde{T}$ enables us to take advantage of Lemmas~\ref{lemma:p} and~\ref{lemma:D_k} with only a loss of an additional constant factor. 

Given a schedule $\mathcal{S}$ for a DAG $G$, we use the same notation as in Section \ref{section:weighted}, $G(\mathcal{S}, j)$, to denote subsets of DAG. For each DAG $G(\mathcal{S}, j)$, there is a terminal chain $\tc(\mathcal{S}, j)$ with task $j$ as the ending task in the schedule $\mathcal{S}(j)$. Similarly, define $P(\mathcal{S}, j)$ as a sum of the processing time along the terminal chain $\tc(\mathcal{S}, j)$,
\begin{equation}
    P(\mathcal{S}, j) = \sum_{\task{c_j} \in \tc(\mathcal{S}, j)} \frac{\nodeW{{c_j}}}{\ps{\ma{c_j}}},
\end{equation}
and let $D_k (\mathcal{S}, j)$ denote the total load assigned to machines in group $M_k$ in DAG $G(\mathcal{S}, j)$,
\begin{equation}
    D_k (\mathcal{S}, j) =  \frac{\sum_{\task{j}: j \in G(\mathcal{S}, j), k \in \ga{j} } \nodeW{j}}{s(M_k)}.
\end{equation}

For every DAG $G(\mathcal{S}, j)$ associated with schedule $\mathcal{S}_j$ for $1 \leq j \leq n$, we are able to apply Separation Principle and then combine these inequalities as follows:
\begin{subequations}
\begin{alignat}{2}
\sum_j \weight{j} C_j & \leq \sum_j \weight{j} \left(P(\mathcal{S}, j) + \sum_k D_k (\mathcal{S}, j)\right) + \sum_j \weight{j} C(\mathcal{S}, j). \notag
\end{alignat}
\end{subequations}
Both $P(\mathcal{S}, j)$ and $D_k (\mathcal{S}, j)$ are independent of the communication constraints, which enables us to take advantage of any group assignment rule.

Using the group assignment rule $\gaw{\cdot}$ helps further tighten the bound.  To show this, we first divide the $n$ tasks into $\tlu$ sets based on $\tl(j)$, which can be viewed as a rough estimate of the completion time of task $\task{j}$. For the $\tl$th interval, we define $\js_\tl$ as a set of tasks such that $\tl(j) = \tl$:
\begin{equation*}
    \js_\tl = \{ j: \tl(j) = \tl \}.
\end{equation*}
In this way, we have divided the $n$ tasks into $\tlu$ sets: $\js_1, \js_2, \ldots, \js_\tlu$. 

Next, for $1 \leq \tl \leq \tlu$, we construct a set of feasible solutions for LP \eqref{opt:milp}, $\tilde{x}, \tilde{C}$ and $\tilde{T}$, for every set of tasks in $\js_\tl$, based on the optimal solution of LP \eqref{weighted:milp}, i.e., $x^*$ and $C^*$. Note that $\tilde{x}$ here is the same as in equation \eqref{eq:add_chill}. Since precedence constraints are preserved in constraints of the LPs, we can concatenate these schedules together to obtain a feasible schedule for all of the tasks.

\begin{lemma} \label{lemma:feasibility}
Consider a set of tasks $\js_\tl$ for a fixed $\tl$.  A feasible solution for LP \eqref{opt:milp} is defined by
\begin{subequations}
\begin{alignat}{3}
\tilde{x}_{i,j} & = \sum_{t=1}^{\tl} \frac{x^*_{i,j,t}}{\alpha_j} & \quad \forall i, j \in \js_\tl \\
\tilde{C}_j & = 2 C^*_j & \forall j \in \js_\tl \\
\tilde{T} & = 2^{\tl+1}. &
\end{alignat}
\end{subequations}
\end{lemma}

\begin{proof}
To show feasibility of such a candidate solution, we verify that $\tl$, $\tilde{x}, \tilde{C} \text{ and } \tilde{T}$ satisfy all the constraints in LP \eqref{opt:milp}. Substitute $\tilde{x}$ into the left side of constraint \eqref{cons:assigned} for any task $\task{j} \in \js_\tl$, and it is clear that $\sum_i \tilde{x}_{i,j} = 1$. To validate that constraint \eqref{cons:cap} is satisfied, note that $\alpha_j \geq 1/2$ by definition and so a direct substitution on the left hand side yields the right hand side due to \eqref{cons:weighted_cap}. Similarly, constraint \eqref{cons:weighted_prec} ensures that constraint \eqref{cons:prec} is satisfied and constraint \eqref{cons:index_cap} ensures that constraint \eqref{cons:cap_each_machine} is satisfied. Finally, we obtain $C^*_j \leq 2^\tl$ by definition of $\tl(j)$ and thus constraint \eqref{cons:completion_time} holds.
% \Halmos
\end{proof}

Due to the similarity between group assignment rule $\gam{\cdot}$ and $\gaw{\cdot}$, we can further tighten the bound using Lemmas \ref{lemma:p} and \ref{lemma:D_k} from Section \ref{section:main_proof} directly. Combining Lemmas \ref{lemma:p} and \ref{lemma:feasibility}, we conclude that the total load along any chain $\tc$ in the DAG formed by $\js_\tl$ is upper bounded by 
\begin{alignat}{2}
    \sum_{\task{j} \in \tc} \frac{\nodeW{{j}}}{\ps{\ma{j}}} & \leq 2 \gamma \tilde{T} \notag \\
    & = 2 \gamma \cdot 2^{\tl+1}. \notag
\end{alignat} 
Next, since the terminal chain $\tc(\mathcal{S}, j)$ can be represented as a concatenation of chains in the DAGs formed by tasks in $\js_\tl$ for $1 \leq \tl \leq \tl(j)$, we have
\begin{alignat}{2}
    P(\mathcal{S}, j) & \leq \sum_{t=1}^{\tl(j)} 2 \gamma \cdot 2^{t+1} \notag \\
    & \leq 8\gamma \cdot 2^{\tl(j)}. \notag
\end{alignat}
Using Lemmas \ref{lemma:D_k} and \ref{lemma:feasibility} together gives the following inequality:
\begin{alignat}{2}
    \sum_{\task{j} \in \js_\tl} \frac{\nodeW{j}}{s({\gaw{j}})} & \leq 2K \tilde{T} \notag \\
    & = 2K\cdot 2^{\tl+1}. \notag
\end{alignat}
The left side can be viewed as $\sum_k D_k$ for a DAG formed by tasks in $\js_\tl$. Since the tasks in DAG $G(\mathcal{S},j)$ form a subset of $\cup_{t=1}^{\tl(j)} \js_\tl$, the following inequality holds:
\begin{subequations}
\begin{alignat}{2}
    \sum_k D_k (\mathcal{S}, j) & \leq \sum_{t=1}^{\tl(j)} \sum_{\task{j^\prime} \in \js_\tl} \frac{\nodeW{j^\prime}}{s({\gaw{j^\prime}})} \notag \\
    & \leq \sum_{t=1}^{\tl(j)} 2K \cdot 2^{t+1} \notag \\
    & \leq 8K \cdot 2^{\tl(j)}, \notag
\end{alignat}
\end{subequations}
which immediately yields
\begin{subequations}
\begin{alignat}{2}
P(\mathcal{S}, j) + \sum_k D_k (\mathcal{S}, j) & \leq 8 (\gamma + K) \cdot 2^{\tl(j)} \notag.
\end{alignat}
\end{subequations}

Finally, the remaining piece of the proof is to upper bound $2^{\tl(j)}$ with a multiplicative factor of its optimal completion time $C^*_j$ in the LP \eqref{weighted:milp}. By definition of $\tl(j)$, for task $\task{j}$ either 
\begin{equation} \label{ineq:1}
    \sum_{t=1}^{\tl(j)-1} \sum_i x^*_{i,j,t} < \frac{1}{2}
\end{equation}
or 
\begin{equation} \label{ineq:2}
    C^*_j > 2^{\tl(j)-1}.
\end{equation}
If inequality \eqref{ineq:1} holds, then
\begin{subequations}
\begin{alignat}{2}
2^{\tl(j)-1} & = \tau_{\tl(j)-1} \notag \\
& \leq 2 \tau_{\tl(j)-1} \left(\sum_{t=\tl(j)}^{\tlu} \sum_i x^*_{i,j,t}\right) \label{ineq:qj}\\
& \leq 2  \left(\sum_{t=\tl(j)}^{\tlu} \tau_{t-1} \sum_i x^*_{i,j,t}\right) \notag \\
& \leq 2  \left(\sum_{t} \tau_{t-1} \sum_i x^*_{i,j,t}\right) \notag\\
& \leq 2 C^*_j \label{ineq:left}.
\end{alignat}
\end{subequations}
Inequality \eqref{ineq:qj} is due to \eqref{ineq:1} and the definition of $\tl(j)$, and constraint \eqref{cons:left_index} in the LP \eqref{weighted:milp} leads to \eqref{ineq:left}.
If inequality \eqref{ineq:2} is true, then
\begin{equation*}
  2^{\tl(j)-1} < C^*_j \leq 2 C^*_j.
\end{equation*}
In both cases, $2^{\tl(j)-1}$ is upper bounded by $2 C^*_j$. Thus, we achieve
\begin{equation*}
    P(\mathcal{S}, j) + \sum_k D_k (\mathcal{S}, j) \leq 32 (\gamma + K) \cdot C^*_j.
\end{equation*}
Since $\sum_j \weight{j} C^*_j$ is lower bounded by $\wopti$, we conclude that
\begin{subequations}
\begin{alignat}{2}
\sum_j \weight{j} C_j & \leq \sum_j \weight{j} \left(P(\mathcal{S}, j) + \sum_k D_k (\mathcal{S}, j)\right) + \sum_j \weight{j} C(\mathcal{S}, j) \\
& \leq 32(\gamma+K)\sum_j \weight{j} C^*_j + \sum_j \weight{j} C(\mathcal{S}, j) \\
& \leq \bigO{(\loglogm)} \cdot \wopti + \sum_j \weight{j} C(\mathcal{S}, j),
\end{alignat}
which completes the proof.
\end{subequations}

\section{Concluding Remarks}
This paper studies the problem of scheduling tasks with precedence constraints on related machines with machine-dependent communication times, and addresses two long-standing open problems in the area.  We introduce a new scheduler, GETF, and prove worst-case approximation ratios for it in the case of (i) scheduling to minimize the makespan and (ii) scheduling to minimize the total weighted completion time.  These results represent the first progress on this problem in the 30 years since \cite{hwang1989scheduling} provided a bound on the makespan under ETF in the case of identical servers and communication time. No previous bounds exist for the case of total weighted completion time when communication time is considered.

A variety of open questions are raised by the work in this paper.  Most importantly, while we have provided theoretical bounds on the performance of GETF, it is also important to investigate how GETF performs in real settings via an implementation study.  GETF could be particularly powerful in the context of large-scale machine learning platforms, where workflows are typically specified as DAGs.  As part of such a study, it would be interesting to understand how to best choose a tie-breaking rule, how to adjust the group assignment rules for the best performance, and how various choices for these rules compare with heuristics that have been suggested in the literature. Further, it will be important to see if it is possible to obtain some theoretical results characterizing how the optimal choices for these rules depend on properties of real-world workloads. Moreover, it will also be interesting to extend the results of this work to stochastic settings, e.g., when task sizes are unknown.

On the analytic side, it will be interesting to discover other applications of the Separation Principle.  It may be possible to revisit other scheduling problems for precedence-constrained tasks and obtain more general results because of the separation this result provides.  Further, it is possible to consider other performance measures, such as energy usage and resource augmentation, using the Separation Principle.

\bibliographystyle{unsrt}  
\bibliography{references}  %%% Remove comment to use the external .bib file (using bibtex).


% George Kour and Raid Saabne.
% \newblock Real-time segmentation of on-line handwritten arabic script.
% \newblock In {\em Frontiers in Handwriting Recognition (ICFHR), 2014 14th
%   International Conference on}, pages 417--422. IEEE, 2014.

% George Kour and Raid Saabne.
% \newblock Fast classification of handwritten on-line arabic characters.
% \newblock In {\em Soft Computing and Pattern Recognition (SoCPaR), 2014 6th
%   International Conference of}, pages 312--318. IEEE, 2014.

% \bibitem{hadash2018estimate}
% Guy Hadash, Einat Kermany, Boaz Carmeli, Ofer Lavi, George Kour, and Alon
% \newblock Estimate and replace: A novel approach to integrating deep neural
%   networks with existing applications.
% \newblock {\em arXiv preprint arXiv:1804.09028}, 2018.

% \end{thebibliography}

\end{document}